\documentclass[8pt]{article}

\usepackage{amsfonts}
\usepackage{amsmath}
\usepackage{relsize}
\usepackage{amssymb}
\usepackage{amsthm}
\usepackage{caption}
\usepackage{authblk}
\newcommand{\scri}{\mathscr{I}}
\usepackage{mathrsfs}
\usepackage{stmaryrd}
\usepackage{mathtools}
\mathtoolsset{showonlyrefs}

\def\Hess{\mbox{Hess}}
\def\grad{\mbox{grad}}
\def\F{\mathcal{F}}

\def\Z{T^{\perp}}
\def\V{T^{\perp\parallel}}
\def\T{T}
\def\Tperpu{\T^{\parallel_u}}
\def\Tperp{\T^{\parallel}}

\usepackage{stmaryrd}

\def\n{\hat{n}}
\def\amu{\alpha\beta\mu\nu}
\def\anu{\alpha\mu\nu\beta}

\def\tr{\mbox{tr}}
\def\traceh{\tr_{\proju}}

\def\g{\bm{g}}
\def\proju{\gamma}
\def\hu{\eta}
\def\Riem{\mbox{Riem}}
\def\Ric{\mbox{Ric}}
\def\circRic{\widehat{\mbox{Ric}}}
\def\Sch{\mbox{Sch}}
\def\Scal{\mbox{Scal}}
\def\Weyl{\mbox{Weyl}}

\def\Bac{\mbox{Bach}}

\def\fop{\mathcal{L}_{\partial_\Om}}

\def\R{\mathcal{R}}
\def\WT{\mathcal{W}}

\newcommand{\cir}[1]{\widehat{#1}}
\newcommand{\bm}[1]{\mbox{\boldmath $#1$}}
\newcommand{\per}[1]{{#1}^{\perp}}
\newcommand{\perpara}[1]{{#1}^{\perp \parallel}}
\newcommand{\para}[1]{{#1}^{\parallel}}

\usepackage{xcolor}
 	\definecolor{carblue}{rgb}{0.30, 0.50, 1}

\newcommand{\redsize}[1]{\text{\scriptsize $#1$}}

\def\bmu{\bm{u}}
\def\u{u}

\def\U{\mathcal{U}}

\def\Scalsig{\Scal_{\Sigma}}
\def\Riemsig{\Riem_{\Sigma}}
\def\Ricsig{\Ric_{\Sigma}}
\def\Rsig{R^{\Sigma}}

\def\K{K}
\def\H{H}
\def\circK{\sigma}

\def\Rnn{R^{\perp }}

\def\Rnt{R^{\perp \parallel}}
\def\Rtt{R^{\parallel}}

\def\circT{\cir{R^{\perp\parallel}}}
\def\circQ{\cir{\Ric(R^{\parallel})}}

\def\trT{{\Rnt_{\mbox{\tiny tr}}}}
\def\circP{\cir{R^{\perp}}}

\def\trP{{\Rnn_{\mbox{\tiny tr}}}}
\def\Q{\Scal(R^{\parallel})}

\def\trZZ{{\Z_{\mbox{\tiny tr}}}}
\def\trV{{\V_{\mbox{\tiny tr}}}}

\def\circV{\widehat{{T}{}^{\perp\parallel}}}
\def\circZZ{\widehat{{T}{}^{\perp}}}

\def\I{\mathrm{I}}
\def\Ia{\mathrm{I}_{\mathrm{a}}}
\def\II{\mathrm{II}}

\def\wand{k}
\def\snlight{l}
\def\snspace{m}

\def\h{h} 
\def\Om{\Omega}

\def\y{y}

\def\Wi{W}
\def\Wper{W^{\perp}}
\def\Wpar{W^{\parallel}}
\def\Wppar{W^{\perp \parallel}}

\def\Cot{\mbox{Cot}}
\def\f{f}
\def\Wperscr{\mathcal{W}^\perp}

\def\obs{\mathcal{O}}
\def\obspol{\mathcal{P}}

\def\projuscr{{\proju_0}}

\def\gn{g_{(\n -1)}}
\def\at{a_{\n}(\projuscr)}
\def\bt{b_{\n}(\projuscr)}

\def\L{\ell}

\def\nabscr{D}
\def\w{w}
\def\ay{b}

\def\eqscr{\stackrel{\scri}{=}}

\newcounter{mnotecount}[section]

\renewcommand{\themnotecount}{\thesection.\arabic{mnotecount}}
\newcommand{\mnote}[1]
{\protect{\stepcounter{mnotecount}}$^{\mbox{\footnotesize
$
\bullet$\themnotecount}}$ \marginpar{
\raggedright\tiny\em
$\!\!\!\!\!\!\,\bullet$\themnotecount: #1} }

\usepackage{amsfonts}
\usepackage{amsmath}
\usepackage{amssymb}

 \newtheorem{theorem}{Theorem}
\newtheorem{lemma}{Lemma}
\newtheorem{proposition}{Proposition}
\newtheorem{corollary}{Corollary}
\newtheorem{remark}{Remark}

\newtheorem{definition}{Definition}

\title{Classification of $\Lambda \neq 0$-vacuum algebraically special spacetimes with conformally flat $\mathscr I$ from Weyl tensor expansion  }

\author[1]{Marc Mars\footnote{marc@usal.es}}
\author[2]{Carlos Peón-Nieto\footnote{carlos.peon@upm.es}}

\affil[1]{Departamento de Física Fundamental, Universidad de Salamanca, Salamanca, Spain.}
\affil[2]{Departamento de MATIC, Universidad Politécnica de Madrid, Madrid, Spain.}

\begin{document}


 \maketitle



\begin{abstract}
 We introduce a general algebraic decomposition of Riemann-like and Weyl-like tensors with respect to a non-null vector $\u$. We derive Gauss, Codazzi and Ricci-type identities for the Weyl tensor, that allow to relate the components of the spacetime Weyl tensor with intrinsic quantities of the hypersurfaces orthogonal to $\u$.
 Restricting to the case of $\Lambda$-vacuum spacetimes (with $\Lambda \neq 0$ and any dimension)  admiting a conformal compactification, we then study the behaviour of the Weyl tensor near $\mathscr I$ by means of an asymptotic expansion {\it \`a la} Fefferman-Graham, where the first terms are explicitly computed. We use these tools to characterize four dimensional algebraically special spacetimes with locally conformally flat $\mathscr{I}$, showing they match exactly the so-called {\it Kerr-de Sitter-like class with conformally flat $\scri$}, thus providing a  geometric characterization of this class of spacetimes.
\end{abstract}

\section{Introduction}

The study of the asymptotics of vacuum spacetimes with a cosmological constant ($\Lambda$) has attracted growing interest over the past decades, motivated by a variety of physical and theoretical developments. Of fundamental importance is the observational milestone of gravitational wave detection by the LIGO and Virgo collaborations \cite{ligovirgo}. Within the paradigm of a universe with $\Lambda > 0$, this experimental achievement does not yet fully match a complete theoretical understanding. Indeed, despite significant progress, there remains no consensus on how to define gravitational radiation in the presence of a non-zero cosmological constant. At the theoretical level, the natural place to study gravitation radiation is future (or past) null infinity, represented in terms of a conformal boundary $\scri$.



A second long-standing motivation involving asymptotic analysis in the case $\Lambda < 0$ is the AdS/CFT correspondence conjecture \cite{maldacena}. Its counterpart with $\Lambda > 0$ (e.g. \cite{stormingerdsCFT}), although less popular as a framework for quantum gravity, makes use of a similar mathematical toolbox that has proven useful in a variety of contexts. For example, the asymptotic formal expansions developed by Fefferman and Graham \cite{FeffGrah85,ambientmetric} apply for all non-zero values of $\Lambda$. However, only in the case $\Lambda > 0$, this machinery allows for a well-posed asymptotic initial value problem in all spacetime dimensions \cite{Anderson2005,andersonchrusciel05,hintz23,kaminski21,kichenassamy03}, extending earlier results \cite{friedrich81,friedrich81bis,Fried86initvalue} in four dimensions by  Helmut Frierdrich. As we shall later comment on, the asymptotic initial data of this Cauchy problem play a central role for some of the proposals for describing gravitational radiation.


The asymptotic region of spacetimes is typically analyzed using conformal extensions {\it \`a la} Penrose. The conformal invariance of the Weyl tensor makes this a very relevant object in the contexts discussed above. In the asymptotic initial value problem in four dimensions, the electric part of the Weyl tensor is one piece of the freely specifiable data. This remains true in higher dimensions when the conformal boundary is locally conformally flat \cite{marspeondata21}. The remaining piece of the data at $\scri$ is the conformal class of boundary metrics. In spacetime dimensions greater than $4$, this is partially characterized by its intrinsic Weyl tensor, which as we shall prove in this paper (cf. equations \eqref{eqweylpar4d}-\eqref{Wparngeq5}), appears explicitly in the expansion of the spacetime Weyl tensor.

On the other hand, in the context of gravitational radiation, we highlight the framework developed by Wald and Zoupas \cite{waldzoupas}, based on the covariant phase space formalism \cite{leewald}. They propose a general prescription for defining and computing the flux of ``energy'' across a boundary, which applies in any dimension and for a broad class of gravitational theories. Interestingly, in the case of the Einstein-Hilbert action with $\Lambda = 0$, this yields a gravitational flux formula at $\scri$ in terms of the news tensor. For $\Lambda \neq 0$, using the Einstein-Hilbert action with holographic counterterms, leads to a flux expression written in terms of the electric part of the Weyl tensor \cite{anninos,compere1,compere2,poole,hoquekrtouspeon},  which recall, is part of the asymptotic initial data.
This approach provides a solid mathematical structure that encompases and generalizes other proposals, such as those by Ashtekar and Magnon \cite{ashtekarmagnon} in four dimensions and their extension to higher dimensions by Ashtekar and Das \cite{ashtekardas}.

It is also worth mentioning a more recent criterion for detecting gravitation radiation proposed by  Fernández-Álvarez and Senovilla \cite{fransenoradiation,fransenoasymptstructure}. In this case, this depends on the interplay between the electric Weyl tensor and the Cotton tensor of the boundary metric. We note that the latter also appears in the expansion of certain components of the spacetime Weyl tensor.






The works mentioned above are just a few among many that have contributed to establish the idea that the Weyl tensor encodes relevant asymptotic degrees of freedom, regardless of the spacetime dimension or the sign of the cosmological constant. In order to make this statement more precise, a careful analysis of the fall-off of the Weyl tensor is central. There exist classical results along these lines by Bondi \cite{bondi1} and Sachs \cite{sachs1,sachs2} (see also Newman and Penrose \cite{newmanpenrose}), analyzing the fall-off of the Weyl tensor for vacuum spacetimes in terms of the affine parameter of a null geodesic field, which became to be known as the {\it peeling} of the Weyl tensor.
These results were subsequently generalized by Penrose to include an arbitary cosmological constant \cite{penrosegravitation}.
A more recent work \cite{peelingsenofran} unifies earlier four-dimensional results for all values of $\Lambda$ within a single, frame-independent, and geometrically meaningful framework.
Additionally, higher dimensional generalizations have been explored, notably in the work by Ortaggio and Pravdová \cite{ortaggioweyl} studying the general asymptotic behaviour of the Weyl tensor in all dimensions (for any value of $\Lambda$), or in the case of Kerr-Schild spacetimes with Minkowski \cite{ortaggiokS} and (anti)-de Sitter \cite{malekpravda11} backgrounds.



In this paper we propose a different perspective for a fall-off analysis of the arbitrary dimensional Weyl tensor  for any value of non-zero $\Lambda$.
We aim at expressing the components of the spacetime Weyl tensor in terms of intrinsic quantities to a submanifold $\Sigma$ defined by orthogonality to a non-null vector $\u$. To that aim, the first step is an algebraic decomposition  that we carry out in Section \ref{secalgebraic} for Riemann and Weyl-like tensors, together with their related Ricci and Schouten-like traces. This decomposition consists in  expressing, in a covariant way, these  tensors in terms of its components projected onto $\Sigma$. We choose carefully such components so that the final result is well adapted to Gauss, Codazzi and Ricci type identities.

The second step is given in Section \ref{secgausscod}, where we particularize the results of Section \ref{secalgebraic} to the Weyl and Riemann tensors obtained from a differentiable spacetime metric $g$. The choice of elements involved in the decomposition in Section \ref{secalgebraic} is central to obtain a set of Gauss, Codazzi and Ricci-like identities for the Weyl tensor and Riemann tensor (cf. Proposition \ref{GCR-decomposed}), which imply analgous results for the Schouten tensor (cf. Remark \ref{remarkWeylg}.)
These identities have the additional advantage of being fully covariant and expressed entirely in terms of quantities intrinsic to $\Sigma$, all of which are either directly derivable from the induced metric $\gamma$ or can be eliminated by an appropriate choice of the vector $\u$.

This turns out to be key in Section \ref{secFG}, were we finally carry out the asymptotic expansion of the Weyl tensor of $(\Lambda \neq 0)$-vacuum spacetimes admitting a conformal extension. For this part, we take advantage of the Fefferman-Graham expansion of the metric $\gamma$ near $\scri$. Adapting $\u$ to the gradient of the conformal factor $\Om$ used in the Fefferman-Graham expansion, the identities in Section \ref{secgausscod} provide an expression of the Weyl tensor depending only on the intrinsic metric $\gamma$. Then using the Fefferman-Graham expansion of $\gamma$ we can readily obtain the asymptotic expansion of the Weyl tensor, computing its first terms explicitly.





The results described so far are general and constitute a mathematical toolbox with many potential applications. We work out a non-trivial application in Section \ref{secalgspec}, which is of particular interest to us. We first write down the algebraically special condition of the Coley {\it et al.} \cite{coley04} classification of Weyl tensors in all dimensions in terms of  the components of the Weyl tensor that we expanded in Section \ref{secFG}. Then, restricting to four dimensions, we insert the expansion in Section \ref{secFG} and evaluate the first and second order terms at $\scri$. The outcome  are two differential equations relating intrinsic objects at $\scri$ (cf. \eqref{WpertypeIIO1} and \eqref{2ndOfull}), directly derivable from the asymptotic initial data. We show that these equations can be integrated whenever $\scri$ is locally conformally flat. The general solution gives a class of asymptotic initial data that
is known to define \cite{marspaetzseno16} an explicit class of metrics known as Kerr-de Sitter-like with conformally flat $\scri$, and which  appeared naturally in a classification of Kerr-de Sitter metrics among $\Lambda$-vacuum spacetimes admiting a Killing vector
(for a similar classification with arbitrary $\scri$ see \cite{Mars_scri1} and for an extension to all dimensions see  \cite{marspeonKSKdS21}.) As a consequence, we obtain an explicit, local, asymptotic classification of all algebraically special metrics admitting a locally conformally flat $\scri$.

We remark that the connection between Kerr-de Sitter spacetimes and algebraically special metrics in {\it five dimensions} was already noted in \cite{reall15}, where the authors characterized Kerr-de Sitter and related metrics as the set of algebraically special spacetimes with a non-degenerate optical matrix. As pointed out in \cite{reall15}, this classification does not extend to four dimensions, although extensions to higher dimensions are possible, for instance, the one carried out in six dimensions in \cite{kokoskaortaggio}. Interestingly, the equivalence between the classification in \cite{reall15} and the Kerr–de Sitter-like class of metrics in five dimensions was established in \cite{marspeonCKVs}.

The results of the present work show that four dimensional, algebraically special spacetimes admitting a locally conformally flat  $\scri$ coincide with the Kerr-de Sitter-like class of metrics. Remarkably, this characterization can be shown to extend to all dimensions, including four, where the previous classification based on the non-degeneracy of the optical matrix fails to apply (see arguments in the original work \cite{reall15}).
This highlights local conformal flatness of $\scri$ as a robust and geometrically meaningful criterion for identifying Kerr-de Sitter and other related spacetimes. The extension to higher dimensions will be presented in a future paper.


\subsection{Setup and notation}

Throughout this paper $(M,g)$ denotes a smooth $\n$-dimensional connected manifold
with a semi-riemannian metric $g$ of arbitrary signature and $(V,\g)$ is a semi-riemannian vector space of dimension $\n$. We always assume  $\n \geq 3$.  The signatures of $g$ and $\g$ are kept arbitrary at the beginning of the paper, but  will be restricted to Lorentzian later on, where we will also take $(V,\g) $ to be the tangent space  of $M$ at a point $p \in M$, namely $(V,\g) = (T_p M , g|_p)$. The associated contravariant metrics are denoted by
$g^{\sharp}$ and $\g^{\sharp}$ respectively.

We use both index-free and abstract index notation at our convenience.
Indices are lowered and raised with $g_{\alpha\beta}$ and its inverse, or
$\g_{\alpha\beta}$ and $\g^{\alpha\beta}$ depending on the context.
As usual, square brackets denote antisymmetrization and round parenthesis symmetrization.
In index-free notation the one-form metrically related to a vector $X$ is denoted $\bm{X}$ and vice versa.
The set of symmetric (0,2) tensors in $V$ is denoted by $S^{(2)}(V)$. The symmetrized tensor product of two-covariant tensors $A$ and $B$ is $A \otimes_s B := \frac{1}{2} (A \otimes B + B \otimes A)$ and
$\pounds_X$ denotes the Lie derivative along a vector field $X$.

\section{Algebraic properties of Riemann-type and Weyl-type tensors}\label{secalgebraic}

A $(0,4)$-tensor $T_{\amu}$ on $V$ is called a {\it Riemann-type tensor} if it has the same algebraic symmetries of a Riemann tensor, namely
\begin{align*}
T_{\amu} = -  T_{\alpha\beta\nu\mu}, \qquad
T_{\amu} = -  T_{\beta\alpha\mu\nu}, \qquad
T_{\amu} =  T_{\mu\nu\alpha\beta}, \qquad
T_{\alpha\beta\mu\nu}  + 
T_{\alpha\mu\nu\beta}  + T_{\alpha\nu\beta\mu}  =0.
\end{align*}
The set of Riemann type tensors defines a vector space which we denote $\R(V)$.
A tensor $T$ is  a {\it Weyl-type tensor} if $T \in \R(V)$ and it is trace-free in the first and third indices (equivalently, trace-free in every pair of indices). The set of Weyl-type tensors is denoted $\WT(V)$.

A standard method to construct tensors $T \in \R(V)$ is by means of the
Kulkarni-Nomizu product. This operation, denoted by $\owedge$ takes a pair
$A,B \in S^{(2)}(V)$  and defines the $(0,4)$-tensor
\begin{align}
  (A \owedge B)_{\amu} := A_{\alpha\mu} B_{\beta\nu} - A_{\alpha\nu} B_{\beta\mu}
  -A_{\beta\mu} B_{\alpha\nu} + A_{\beta\nu} B_{\alpha\mu}. \label{KN}
\end{align}
It is immediate to check that $A \owedge B \in \R(V)$.

We shall need a second operation that constructs tensors $T \in \R(V)$. 
As usual, we call $\Lambda^{(p)}(V)$ the vector space of $p$-forms on $V$ and introduce $\Lambda^{(1,2)} (V)$ as the set of 
$(0,3)$-covariant tensors $F$ in $V$ which are antisymmetric in the last two indices and satisfy the condition
\begin{align*}
F_{\alpha\mu\nu} +  F_{\mu\nu\alpha} + F_{\nu\alpha\mu} =0.
\end{align*}
Given $\bm{u} \in \Lambda^{(1)} (V)$ and
$F \in \Lambda^{(1,2)}(V)$ we define
\begin{align}
  ( \bm{u} \circledast F)_{\amu} := u_{\alpha} F_{\beta\mu\nu}
  -  u_{\beta} F_{\alpha\mu\nu} + u_{\mu} F_{\nu\alpha\beta} - u_{\nu} F_{\mu\alpha\beta}. \label{circledast}
\end{align}
It is straightforward to show that $u \circledast F \in \R(V)$. When $\bm{u}$
has non-zero norm, it can be scaled to become unit $\g^{\sharp}(\bm{u}, \bm{u}) = \epsilon$, with
$\epsilon = \pm 1$.  Let $u$ be the metrically associated vector. The space $\Pi_u$ of vectors orthogonal to $u$ is an $(\n -1)$-dimensional vector space that
inherits a non-degenerate metric from $\g$. Moreover $V = \mbox{span} (u) \oplus
\Pi_u$. The projector $\proju : V \rightarrow \Pi_u$ is given explicitly by
\begin{align*}
  \proju := \mbox{Id} - \epsilon u \otimes \bm{u},
\end{align*}
where $\mbox{Id}$ is the identity tensor on $V$.
Covariant tensors of $V$ which annihilate $u$ are in one-to-one correspondence with covariant tensors on $\Pi_u$. We call such tensors {\it completely orthogonal to $u$} and for them one can raise and lower indices indistinctly with $\g^{\alpha\beta}$ or with $\proju^{\alpha\beta} = \g^{\alpha\beta} - \epsilon
  u^{\alpha} u^{\beta}$. Similarly, traces in any pair of indices can be defined indistinctly with $\g$ or with $\proju$. We shall use $tr_{\proju}$ to denote this trace.

The following lemma gives a canonical decomposition of any tensor $T \in \R(V)$
with respect to the vector $u$.
    \begin{lemma}
  \label{decom1}
    Let $(V,\g)$ have dimension $\n \geq 4$ and
  fix a unit vector $u \in V$
  with square norm $g(u,u)=\epsilon \in \{ -1,1\}$. Let $\T \in \R (V)$. Then, there exists a unique decomposition
\begin{align}
  \T = \Tperpu + \epsilon \bm{\u} \circledast \V + \Z \owedge \hu. \label{decomu}
\end{align}
where
  \begin{align*}
\hu := \bm{u} \otimes \bm{u} - \frac{\epsilon}{\n-3} \proju,
 \end{align*}
and $\Tperpu \in \R (V)$, $\V \in \Lambda^{(1,2)}(V)$, $\Z \in S^{(2)}(V)$ are completely orthogonal to $u$ and have explicit expressions
\begin{align}
  \Z_{\beta\nu} & := \u^{\alpha} \u^{\mu} \T_{\alpha\beta\mu\nu} \label{defZ} \\
  \V_{\beta\mu\nu} & := \u^{\alpha} \proju^{\kappa}_{\mu} \proju^{\sigma}_{\nu} \T_{\alpha\beta\kappa\sigma} \label{defV} \\
  \Tperpu & := \Tperp + \frac{\epsilon}{\n-3}
                   \Z \owedge \proju. \label{defTperp}
                   \qquad \mbox{with} \qquad
\Tperp_{\amu} :=  \proju^{\rho}_{\alpha} \proju^{\sigma}_{\beta} \proju^{\kappa}_{\mu}
                   \proju^{\delta}_{\nu} \T_{\rho\sigma\kappa\delta}.
\end{align}
Moreover $T \in \WT(V)$ if and only if 
  \begin{align*}
    \Tperpu \in \WT(V), \qquad \V{}^{\alpha}{}_{\alpha\mu} =0,
    \qquad \traceh \Z =0.
  \end{align*}
\end{lemma}

\begin{proof} 
  Define $\Z$ and $\V$ by means of \eqref{defZ}-\eqref{defV}. By the symmetries of $T$ these tensors also admit the expressions
  \begin{align*}
    \Z_{\beta\nu}  := \u^{\alpha} \proju^{\rho}_{\beta} \u^{\mu} 
                    \proju^{\kappa}_{\nu} \T_{\alpha\rho\mu\kappa}, \qquad \quad
  \V_{\beta\mu\nu}  := \u^{\alpha} \proju^{\rho}_{\beta} \proju^{\kappa}_{\mu} \proju^{\sigma}_{\nu} \T_{\alpha\rho\kappa\sigma},
\end{align*}
  so they are completely orthogonal to $u$ and satisfy $\Z \in S^{(2)}(V)$, $\V \in \Lambda^{(1,2)}(V)$. Define the tensor $\Tperpu$ by means
  of \eqref{decomu}. By construction we have
  $\Tperpu \in \R(V)$. Contracting \eqref{decomu} with
$u$ in the first index and using 
$\hu_{\alpha\beta} u^{\alpha} = \epsilon u_{\beta}$ yields
\begin{align*}
  \T_{\amu} u^{\alpha}
  = \Tperpu_{\amu} u^{\alpha} + \V_{\beta\mu\nu} + \epsilon u_{\mu} \Z_{\beta\nu}
  - \epsilon u_{\nu} \Z_{\beta\mu}.
\end{align*}
On the other hand
\begin{align*}
  \T_{\amu} u^{\alpha} & =
  \T_{\alpha\beta\kappa\sigma} u^{\alpha} 
  \delta^{\kappa}_{\mu}
  \delta^{\sigma}_{\nu} 
 =
\T_{\alpha\beta\kappa\sigma} u^{\alpha} 
\left (  \proju^{\kappa}_{\mu} + \epsilon u^{\kappa} u_{\mu} \right )
\left ( \proju^{\sigma}_{\nu}  + \epsilon u^{\sigma} u_{\nu} \right ) \\
& =
 \V_{\beta\mu\nu} + \epsilon u_{\mu} \Z_{\beta\nu}
- \epsilon u_{\nu} \Z_{\beta\mu}. 
\end{align*}
Combining both expressions we conclude  that $\Tperpu{}_{\amu}$ is  orthogonal to $u^{\alpha}$ and hence
completely orthogonal to $u$. This proves existence of the decomposition.  Since
\begin{align*}
  \proju^{\mu}_{\alpha} \proju^{\nu}_{\beta} \hu_{\mu\nu} = - \frac{\epsilon}{\n-3} \proju_{\alpha\beta}
\end{align*}
expression \eqref{defTperp} follows directly from \eqref{decomu} by projecting all indices orthogonally with respect to $u$. Uniqueness
of the decomposition is immediate because \eqref{defZ}-\eqref{defTperp} follow
from \eqref{decomu} by direct computation. 

It remains to find under which conditions $T \in \WT(V)$.  The trace of
the tensors $\bm{\u} \circledast \V$ and $\Z \owedge \hu$ are
\begin{align*}
  (\bm{\u} \circledast \V)^{\alpha}{}_{\beta\alpha\nu} &=
  -u_{\beta} \V{}^{\alpha}{}_{\alpha \nu} - u_{\nu} \V{}^{\alpha}{}_{\alpha \beta} \\
  (\Z \owedge \hu)^{\alpha}{}_{\beta\alpha\nu} &=
  (\tr_\proju \Z) \hu_{\beta\nu}
  + (\tr_{\redsize{\g}} \hu) \Z_{\beta\nu}
- 2 \hu{}^\alpha{}_{(\beta} \Z_{\nu)\alpha}
= (\tr_\proju \Z) \hu_{\beta\nu}
 \end{align*}
because $\tr_{\redsize{\g}} \hu = - \frac{2\epsilon}{\n -3}$ and $\hu{}^{\alpha}{}_{\beta}
\Z_{\nu\alpha} = - \frac{\epsilon}{\n-3} \Z_{\beta\nu}$.  Thus,
\begin{align*}
\T{}^{\alpha}{}_{\beta\alpha\nu} = 
\Tperpu{}^{\alpha}{}_{\beta\alpha\nu} -
\epsilon u_{\beta} \V{}^{\alpha}{}_{\alpha\nu}
- \epsilon u_{\nu} \V{}^{\alpha}{}_{\alpha\beta}
+ (\tr_\proju \Z) \hu_{\beta\nu}
\end{align*}
and $\T \in \WT(V)$ if and only if $\Tperpu \in \WT(V)$, 
$\V{}^{\alpha}{}_{\alpha\beta} = 0$, $\tr_{\proju} \Z=0$ follows at once.
\end{proof}

Consider a tensor $T \in \R(V)$. There is a canonical method to remove the trace
of $T$ and hence define a tensor $\Weyl(T) \in \WT(V)$. We start by defining
the only independent  trace that a tensor $T \in \R(V)$ can have. Since the trace of the Riemann tensor is the Ricci tensor, we call the corresponding operation
 $\Ric(T)$ and define the symmetric $(0,2)$-tensor
\begin{align*}
\Ric(T){}_{\beta\nu} := T^{\alpha}{}_{\beta\alpha\nu}.
\end{align*}
For a similar reason, the operation of taking trace to this tensor is called $\Scal(T)$
\begin{align*}
\Scal(T) := \tr_{\g} (\Ric(T)).
\end{align*}
Assuming $\n \geq 3$, the operation that extracts the trace to $T$ involves the  symmetric $(0,2)$-tensor defined by
\begin{align}\label{defsch}
\Sch(T) := \frac{1}{\n -2} \left ( \Ric(T)
- \frac{ \Scal (T)}{2 (\n -1)} \g \right )
\end{align}
($\Sch$ stands for Schouten) and is given by
\begin{align}\label{defweyl}
\Weyl(T) := T - \Sch(T) \owedge \g.
\end{align}
 By construction this tensor satisfies $\Weyl(T) \in \WT(V)$. 
Note that $\Weyl(\T)$ is, up to a multiplicative constant, the only linear combination of $\T$ that belongs to $\WT(V)$.

\begin{remark}\label{remarkpuretrace}
 Elements of $T \in \R(V)$ satisfying $\Weyl(T) = 0$ are said to be pure-trace because they are always written $T = A \owedge \g$, for some $A \in  S^{(2)}(V)$, which by definition of $\Sch(T)$ satisfies $A = \Sch(T)$.
\end{remark}

The decomposition in Lemma \ref{decom1} has introduced tensors, such as $\Tperp$ or $T^{\parallel u}$ , that not only have the symmetries of a Riemann-type tensor but, in addition, are
completely orthogonal to $u$. Such tensors can be identified with
elements in $\R(\Pi_u)$, i.e. tensors with the symmetries of a Riemann-type tensor, but living on a space of dimension $\n-1$ instead of $\n$. It is therefore natural to introduce an operation that extracts the trace of the elements in $\R(\Pi_u)$ and leaves a tensor in $\WT(\Pi_u)$. To that aim, we assume $\n \geq 4$
and introduce the following two tensors (note the label $u$ in $\Sch_u$ and
$\Weyl_u$)
\begin{align}
\Sch_u(\Tperp) & := \frac{1}{\n -3} \left ( \Ric(\Tperp) 
- \frac{ \Scal(\Tperp) }{2 (\n -2)} \proju \right ), \label{FirDec} \\
\Weyl_u (\Tperp) & := \Tperp - \Sch_u (\Tperp) \owedge \proju. \label{SecDec}
\end{align}
By construction we have $\Sch_u(\Tperp) \in S^{(2)} (V)$,
$\Weyl_u (\Tperp) \in \WT(V)$ and moreover,
both are completely orthogonal to $u$.

\begin{remark}\label{remarkuniqTperp}
For later use we consider the following uniqueness properties of $\Weyl_u(\Tperp)$. First, $\Weyl_u(\Tperp)$, is the unique, up to a multiplicative constant, linear combination of $\Tperp$ in $\WT(\Pi_u)$.
 On the other hand, the set of elements of $\WT(V)$ completely orthogonal to $u$ coincides with $\WT(\Pi_u)$. Therefore, up to multiplicative constant, \eqref{SecDec}
is the {\it only linear combination} of $T^{\parallel}$ that defines a tensor in $\WT(V)$ completely orthogonal to $u$.  
\end{remark}


We want to use Lemma \ref{decom1} to give an explicit expression of $\Weyl(T)$ 
in terms of $\Weyl_u (\Tperp)$, as well as objects constructed from 
$\Z$ and $\V$. As a by product we shall also get an expression for 
$\Sch(T)$ in terms of such quantities. Since $\Weyl(T)$ is trace-free one can expect that  the trace-free parts of $\Z$ and $\V$ will play a prominent role. 
We shall use a hat to denote trace-free quantities.

Specifically we define
\begin{align}
  \trZZ & := \traceh \Z, \quad \quad 
  &&\circZZ := \Z - \frac{\trZZ}{\n-1}  \proju, \label{P1}\\
  \trV_{\mu} & := \V_{\alpha\beta\mu} \proju^{\alpha\beta}, \quad \quad \qquad
     &&          \circV_{\alpha\beta\mu} := \V_{\alpha\beta\mu} +
               \frac{1}{\n-2} \left ( 
               - \proju_{\alpha\beta} \trV_{\mu} + \proju_{\alpha\mu} \trV_{\beta}  
               \right ). \label{T1}
\end{align}
Note that $\circZZ \in S^{(2)}(V)$, $\circV \in \Lambda^{(1,2)} (V)$ and both are completely orthogonal to $u$ and trace-free.

In order to find the relation between $\Weyl(T)$ and 
$\Weyl_u(\Tperp)$ we shall need to apply the decomposition in Lemma 
\ref{decom1} to the tensor 
$\Sch(T) \owedge \g$. To that end, we recall that any tensor $S \in S^{(2)}(V)$ admits a unique
canonical decomposition in tangential and normal components to $u$, according to
\begin{align*}
S = \para{S} + 2 \epsilon \perpara{S} \otimes_s \bm{u} 
+ \per{S} \bm{u} \otimes \bm{u},
\end{align*}
where 
\begin{align*}
\para{S}{}_{\beta\nu} := \proju^{\rho}_{\beta} \proju^{\kappa}_{\nu}
S_{\rho\kappa}, \qquad
\perpara{S}{}_{\beta} := \proju^{\rho}_{\beta} u^{\kappa} S_{\rho\kappa}, \qquad
\per{S} := u^{\rho} u^{\kappa} S_{\rho\kappa}.
\end{align*}
The next lemma specializes the decomposition of
Lemma \ref{decom1} to the Kulkarni-Nomizu product $A \owedge B$. 
\begin{lemma}
\label{decom2}
Let $A, B \in S^{(2)}(V)$. Then the tensor $A \owedge B$ decomposes according
to Lemma \ref{decom1} as
\begin{align*}
A \owedge B = ( A \owedge B)^{\parallel_u}
+ \epsilon \bm{u} \circledast (A \owedge B)^{\perp \parallel} 
+ (A \owedge B)^{\perp} \owedge \hu
\end{align*}
where
\begin{align}
(A \owedge B)^{\perp} & := 
\per{A} \para{B}
+ \per{B} \para{A}  - 2 \perpara{A} \otimes_s \perpara{B}, \label{Rel1}\\
(A \owedge B)^{\perp \parallel}{}_{\beta\mu\nu} & :=
2 \perpara{A}{}_{[\mu} \para{B}{}_{\nu] \beta}
+ 2 \perpara{B}{}_{[\mu} \para{A}{}_{\nu] \beta}, \label{Rel2}\\
( A \owedge B)^{\parallel_u} & :=
\para{A} \owedge \para{B} 
+ \frac{\epsilon}{\n -3} 
\left ( \per{A} \para{B}
+ \per{B} \para{A}  - 2 \perpara{A} \otimes_s \perpara{B}  \right )  \owedge  \proju. \label{Rel3}
 \end{align}
\end{lemma}

\begin{proof}
  We first note that the full projection of $A \owedge B$ is clearly
  \begin{align*}
    (A \owedge B)^{\parallel} = \para{A} \owedge \para{B}
  \end{align*}
  so \eqref{Rel3} follows from \eqref{defTperp} once we establish \eqref{Rel1}.
  To prove \eqref{Rel1}-\eqref{Rel2} we start by noting that for any symmetric
  tensor $S \in S^{(2)}$ the contraction with $u$ decomposes as
  \begin{align*}
    S_{\beta\nu} u^{\nu} = \perpara{S}_{\beta} + \epsilon \per{S} u_{\beta}.
  \end{align*}
  Inserting  this into the definition \eqref{KN} of the Kulkarni-Nomizu product
    yields
  \begin{align*}
    (A \owedge B)_{\amu} u^{\alpha} =
    2 \left ( \perpara{A}_{[\mu} + \epsilon \per{A} u_{[\mu} \right )
        B_{\nu]\beta}+
      2 \left ( \perpara{B}_{[\mu} + \epsilon \per{B} u_{[\mu} \right )
          A_{\nu]\beta}.
  \end{align*}
  Since $(A \wedge B)^{\perp\parallel}$ is the full projection of this tensor, expression \eqref{Rel2} follows. For \eqref{Rel1}, we simply contract with $u^{\mu}$ and project the remaining indices to get
  \begin{align*}
    (A \owedge B){}_{\alpha\rho\mu\kappa} u^{\alpha} \proju^{\rho}_{\beta} \,
    u^{\mu} \proju^{\kappa}_{\nu}=
    \per{A} \para{B}_{\beta\nu}
    - \perpara{A}{}_{\nu} \perpara{B}{}_{\beta}
  +  \per{B} \para{A}_{\beta\nu}
    - \perpara{B}{}_{\nu} \perpara{A}{}_{\beta}
  \end{align*}
  which is \eqref{Rel1}  (we could also proceed without projecting the remaining indices, the calculation would be somewhat longer, but of course the result is the same).
  \end{proof}
  
\begin{corollary}
  \label{decom3}
Let $A \in S^{(2)}(V)$. Then $A \owedge \g$ decomposes according to
Lemma \ref{decom1} as
\begin{align*}
  ( A \owedge \g) =
  ( A \owedge \g)^{\parallel_u}
  + \epsilon \bm{u} \circledast (A \owedge \g)^{\perp \parallel}
  + (A \owedge \g)^{\perp} \owedge \hu
\end{align*}
where
\begin{align}
  (A \owedge \g)^{\perp} & := \per{A} \proju + \epsilon \para{A}, \nonumber \\
    (A \owedge \g)^{\perp \parallel}{}_{\beta\mu\nu} & :=
  \perpara{A}_{\mu} \proju_{\nu\beta}-
  \perpara{A}_{\nu} \proju_{\mu\beta}, \label{perparaA} \\
  ( A \owedge \g)^{\parallel_u} & :=
  \frac{1}{\n-3} \left (
  (\n-2)  \para{A} + \epsilon \per{A} \proju \right )\owedge \proju. \nonumber
  \end{align}
\end{corollary}

\begin{proof}
Immediate from Lemma \ref{decom2} and the following simple facts
\begin{align*}
\para{\g} = \proju, \qquad \quad
\perpara{\g} =0, \qquad \quad 
\per{\g} = \epsilon.
\end{align*}
\end{proof}

We can now use Lemma \ref{decom1} to
obtain explicit expressions for $\Sch(T)$ and
$\Weyl(T)$ in terms of $\Tperp$, $\V$ and $\Z$, for any tensor
$T \in \R(V)$.
\begin{proposition}
\label{DecomSchWeyl-type}
  In the setup of Lemma \ref{decom1} and with the notation above, assume $\n \geq 4$. Then, the
  following identities
  hold
  \begin{align}
    \Sch(T) = & \frac{1}{\n-2} \left ( \epsilon \cir{T^{\perp}} + \cir{\Ric(T^{\parallel})}
    \right ) + \frac{\Scal(T^{\parallel})}{2 (\n-1)(\n-2)} \proju
                  - \frac{2 \epsilon}{\n-2} \bm{u} \otimes_s  \trV
 + \frac{1}{\n-1} \left ( \trZZ - \frac{\epsilon \Scal(T^{\parallel})}{2 (\n-2)} \right )
    \bm{u} \otimes \bm{u}
    \label{SchoutenDecom1} \\
    \Weyl(T) = & \Weyl_u(T^{\parallel} ) + \epsilon \bm{u}
    \circledast \circV
    + \frac{1}{\n -2} \left ( (\n-3) \circZZ   - \epsilon \cir{\Ric(T^{\parallel})}
    \right ) \owedge \hu. \label{WeylDecom1}
%
%
                                  \end{align}
\end{proposition}

\begin{proof}
  We apply Lemma \ref{decom1} both to the left-hand side and to the right-hand side of
  $\Weyl(T) = T - \Sch(T) \owedge \g$ and use Corollary \ref{decom3} in the
  right-hand side. This yields
  \begin{align}
\Weyl(T)  = &    \Weyl(T)^{\parallel_u} +
\epsilon \bm{u} \circledast \Weyl(T)^{\perp \parallel}
+ \Weyl(T)^{\perp}
\owedge \hu \label{firstline} \\
=  & T^{\parallel_u}  - \frac{1}{\n-3} \left (
  (\n-2)  \para{\Sch(T)} + \epsilon \per{\Sch(T)} \proju \right )\owedge \proju 
+ \epsilon \bm{u} \circledast \left ( T^{\perp \parallel} - \left
( \Sch(T) \owedge \g \right )^{\perp \parallel} \right ) \nonumber
\\ & + \left ( T^{\perp}  -  \Sch(T)^{\perp} \proju - \epsilon
\Sch(T)^{\parallel} \right ) \owedge \hu.  \nonumber 
  \end{align}
  By uniqueness of the decomposition, it follows
  \begin{align}
    \Weyl(T)^{\parallel_u} & = T^{\parallel_u}  - \frac{1}{\n-3} \left (
    (\n-2)  \para{\Sch(T)} + \epsilon \per{\Sch(T)} \proju \right )\owedge \proju \nonumber \\
    & = T^{\parallel}
    - \frac{1}{\n -3} \left ( - \epsilon \T^{\perp}
    + (\n-2)  \para{\Sch(T)} + \epsilon \per{\Sch(T)} \proju \right )\owedge \proju,   \label{Weyl1} \\
    \Weyl(T)^{\perp \parallel} &= T^{\perp \parallel} - \left
( \Sch(T) \owedge \g \right )^{\perp \parallel},  \label{Weyl2} \\
    \Weyl(T)^{\perp} & = T^{\perp}  -  \Sch(T)^{\perp} \proju - \epsilon
\Sch(T)^{\parallel}. \label{Weyl3}
      \end{align}
 Also by Lemma \ref{decom1} we know that
  $\Weyl(T)^{\parallel_u}$, $\Weyl(T)^{\perp \parallel}$ and
  $\Weyl(T)^{\perp}$ must all be trace-free. For \eqref{Weyl2}, and given that
  $( \Sch(T) \owedge \g)^{\perp \parallel}$ has vanishing trace-free part (see \eqref{perparaA}), this means that
\begin{align}
  \Weyl(T)^{\perp \parallel} &= \cir{T^{\perp \parallel}}, \label{Weylperpar} \\
  \left ( \Sch(T) \owedge \g \right )^{\perp \parallel}_{\mbox{\tiny tr}} & \stackrel{\eqref{perparaA}}{=}
(2 - \n) \Sch(T)^{\perp \parallel} \stackrel{\eqref{Weyl2}}{=} \trV, \qquad
\Longleftrightarrow \qquad
\Sch(T)^{\perp \parallel} = - \frac{\trV}{\n-2}. \label{Schperpar}
\end{align}
The condition that the trace in \eqref{Weyl3} vanishes is
 \begin{align}
   \trZZ - (\n-1) \Sch(T)^{\perp}  -
   \epsilon \traceh ( \Sch(T)^{\parallel} ) =0. \label{Cond1}
 \end{align}
 Concerning \eqref{Weyl1}, by Remark \ref{remarkuniqTperp}  \eqref{SecDec} is the only traceless linear combination of $T^{\parallel}$ completely orthogonal to $u$. Thus, equality \eqref{Weyl1} is equivalent to
 \begin{align}
   \Weyl(T)^{\parallel_u} & = \Weyl_u(T^{\parallel}), \label{weylproject}\\
   \Sch_u(T^{\parallel}) & = \frac{1}{\n -3} \left (
   - \epsilon \T^{\perp}
    + (\n-2)  \para{\Sch(T)} + \epsilon \per{\Sch(T)} \proju \right ).\nonumber
 \end{align}
 Inserting the definition \eqref{FirDec} of $\Sch_u(T^{\parallel})$ in the second equality gives
 \begin{align}
 \Ric(\Tperp) 
 - \frac{ \Scal(\Tperp) }{2 (\n -2)} \proju
 = - \epsilon \T^{\perp}
 + (\n-2)  \para{\Sch(T)} + \epsilon \per{\Sch(T)} \proju. \label{perSch}
  \end{align}
 We can decompose $\Ric(T^{\parallel})$  and $\Z$ in terms of its trace and trace-free part as
 \begin{align}
   \Ric(T^{\parallel}) =
   \cir{\Ric(T^{\parallel})} +
   \frac{\Scal (T^{\parallel})}{\n-1} \proju, \qquad \quad
   \Z = \cir{\Z} + \frac{\trZZ}{\n-1} \proju
\label{trace-tracefree}
 \end{align}
 so equation  \eqref{perSch} takes the form
 \begin{align}\label{eqdecsch}
   \cir{\Ric(T^{\parallel})}
   + \frac{(\n -3) \Scal(T^{\parallel})}{2 (\n-1) (\n-2)} \proju
   = - \epsilon \cir{\T^{\perp}} - \frac{\epsilon \trZZ}{\n-1} \proju
   + (\n-2)  \para{\Sch(T)} + \epsilon \per{\Sch(T)} \proju.
 \end{align}
Taking the trace of this equation gives
 \begin{align}
  (\n-2) \traceh ( \Sch(T)^{\parallel})
   + \epsilon (\n-1) \Sch(T)^{\perp}
   - {\epsilon \trZZ} -
   \frac{(\n-3) \Scal (T^{\parallel})}{2  (\n-2)} =0. \label{Cond3}
 \end{align}
 Solving \eqref{Cond1} and  \eqref{Cond3} for $\Sch(T)^{\perp}$ and
 $\traceh ( \Sch(T)^{\parallel})$ yields
 \begin{align}
\Sch(T)^{\perp} = \frac{1}{\n-1} \left ( \trZZ -  
\frac{\epsilon \Scal (\Tperp)}{2 (\n-2)} \right ), \qquad
\traceh ( \Sch(T)^{\parallel}) =  \frac{\Scal (\Tperp)}{2 (\n-2)},
\label{traces}
\end{align}
which inserted back into \eqref{eqdecsch} provides an explicit expression for $\para{\Sch(T)}$
 \begin{align}
   \Sch(T)^{\parallel} = \frac{1}{\n-2} \left ( \epsilon \cir{\Z} + \cir{\Ric(T^{\parallel})}  \right )  + \frac{\Scal (\Tperp)}{2 (\n-1)(\n-2)}\proju.
\label{parallelSch}
 \end{align}
The first claim in the Lemma follows from the decomposition
\begin{align*}
\Sch(T) = \Sch(T)^{\parallel} + 2 \epsilon \bm{u} \otimes_s \Sch(T)^{\perp \parallel}
+ \Sch(T)^{\perp} \bm{u} \otimes \bm{u}
\end{align*}
together with 
\eqref{parallelSch}, \eqref{Schperpar} and \eqref{traces}. For the second part of the lemma observe that, since $\Weyl(T)^{\perp}$ is a traceless tensor, we can remove the trace in the right hand side of \eqref{Weyl3} so that it yields
\begin{align}\label{Weyl4}
\Weyl(T)^{\perp} & = \cir{T^{\perp}}  - \epsilon
\cir{(\Sch(T)^{\parallel})} \stackrel{\eqref{parallelSch}}{=}  \frac{1}{\n-2} \left ( ( \n -3 ) \cir{T^{\perp}} -
\epsilon \cir{\Ric(T^{\parallel})} \right ).
\end{align}
Now \eqref{WeylDecom1} follows by substituting the terms in  \eqref{firstline} according to \eqref{weylproject}, \eqref{Weylperpar} and \eqref{Weyl4}.
\end{proof}

When we deal with tensor fields on $M$ we write $T \in \R(M)$ when
$T|_p  \in \R(T_p M)$ for all $p \in M$. Similar definitions are made for the other types of tensors introduced above.

\section{Decomposition of the Gauss, Codazzi and Ricci identities}\label{secgausscod}

It is well-known that the Riemann tensor  of $(M,g)$  can be expressed in terms of the geometry of the leaves of a
foliation $\{ \Sigma_{\lambda}\} $ of hypersurfaces with non-degenerate first
fundamental form.  Our aim in this section is to derive corresponding expressions for the Weyl tensor  in terms of the geometry of the leaves.

Although our main interest is in foliations by embedded hypersurfaces, there is no increase in complexity in considering leaves
that the are merely injectively immersed, so we work in this setting. By the Fr\"obenius theorem, the foliation can be defined as the integral manifolds of an $(\n-1)$-dimensional distribution $\mathfrak{P}$ defined as the kernel of a smooth nowhere vanishing one-form field $\bmu$. Since by assumption the first fundamental form of each leaf is a metric, $\u$ is nowhere null and we may fix $\bmu$ to be unit $g^{\sharp}(\bmu,\bmu) = \epsilon$,
$\epsilon \in \{ -1,1\}$, which we assume from now on.
The projector orthogonal to $\u^{\alpha}$ is $\proju^{\alpha}_{\beta} = \delta^{\alpha}_{\beta} - \epsilon \u^{\alpha} \u_{\beta}$.  In dimension $\n \geq 4$ we also introduce, as before, the symmetric $(0,2)$-tensor field
\begin{align*}
\hu := \bm{u} \otimes \bm{u} - \frac{\epsilon}{\n-3} \proju.
\end{align*}

We write $\nabla$ for the Levi-Civita derivative of $g$ and
$\Riem_g$ for the corresponding Riemann $(0,4)$-tensor. 
The Ricci tensor  $\Ric_g$ is the trace of $\Riem_g$ (in the first and third indices) and the scalar curvature $\Scal_g$ is the trace of $\Ric_g$. In index notation we write, respectively, $R_{\alpha\beta\mu\nu}$, $R_{\beta\nu}$ and $R$ to denote them, and we also adopt the sign convention
\begin{equation}
      R^\alpha{}_{\beta\mu\nu} X^\beta = \nabla_\mu \nabla_\nu X^\alpha - \nabla_\nu \nabla_\mu X^\alpha.
     \end{equation}
The Schouten tensor of $g$ is
\begin{align}
  \Sch_g  = \frac{1}{\n-2} \left (
  \Ric_g - \frac{\Scal_g}{2 (\n-1)} g \right ) = \frac{1}{\n-2} \circRic_g + \frac{\Scal_g}{2 \n (\n-1)} g.
  \label{def:L}
\end{align}
In index notation we write $L_{\alpha\beta}$ for the Schouten tensor. Note that the
Ricci tensor can be expressed in terms of the Schouten tensor
as
\begin{align}
  \Ric_g = (\n-2) \Sch_g + \tr_g (\Sch_g) g.
  \label{RicSch}
\end{align}
The Weyl tensor of $g$ is
\begin{align}
  \Weyl_g = \Riem_g - \Sch_g \owedge g. \label{Weyl}
\end{align}
which we write  as $C_{\anu}$ in index notation.

As we shall recall later, the expressions of the curvature tensor $\Riem_g$ in terms on the geometry
of the leaves arise via the Gauss, Codazzi and Ricci identities which involve, respectively, the totally tangential, the one normal-three tangential and the two normal-two tangential components of the Riemann tensor. Using the same notation as in the previous section, we therefore introduce the following tensors
\begin{align}
\Rtt_{\alpha\beta\mu\nu}
    &:= \proju_{\alpha}^{\gamma}
      \proju_{\beta}^\delta \proju_{\mu}^\rho \proju_{\nu}^{\sigma}
  R_{\gamma\delta\rho\sigma}, \nonumber \\
    \Rnt_{\beta\mu\nu} & :=\u^{\alpha}  \proju_{\mu}^\rho \proju_{\nu}^{\sigma} R_{\alpha\delta\rho\sigma}, \label{tensorsDecom} \\
  \Rnn_{\beta\nu} & := \u^{\alpha}  \u^{\mu} R_{\alpha\beta\mu\nu}, \nonumber
\end{align}
which are orthogonal to $\u^{\delta}$ in all indices.  All the notation introduced in the previous section applies verbatim to these tensor fields.  Directly from Proposition \ref{DecomSchWeyl-type} we get an explicit decomposition for the Schouten and Weyl tensors of $(M,g)$.
\begin{lemma}
Assume that $(M,g)$ has dimension $\n \geq 4$. Then the 
Schouten tensor  and the Weyl tensor of $g$ admit the following  decomposition
  \begin{align}
        \Sch_g = & \frac{1}{\n-2} \left ( \epsilon \circP + \circQ
    \right ) + \frac{\Q}{2 (\n-1)(\n-2)} \proju
                  - \frac{2 \epsilon}{\n-2} \bm{u} \otimes_s  \trT
 + \frac{1}{\n-1} \left ( \trP - \frac{\epsilon \Q}{2 (\n-2)} \right )
    \bm{u} \otimes \bm{u}
    \label{SchoutenDecom} \\
\Weyl_g & = \Weyl_u (R^{\parallel}) + \epsilon \bm{u} 
\circledast \circT 
+ \frac{1}{\n -2} \left ( (\n-3) \circP   - \epsilon \circQ  )
                              \right ) \owedge \hu. \label{WeylDecom}
%
  \end{align}
   \end{lemma}
\begin{remark}
  This lemma does not require $\bmu$ to be integrable. It only uses that
  $\bmu$ is unit with square norm $\epsilon$.
\end{remark}

\begin{remark}\label{remarkWeylg}
  We quote the following identities for later use. They are all straightforward consequences of the lemma (and actually all except \eqref{r6} have been obtained
  in the course of the proof of Proposition \ref{DecomSchWeyl-type}).
   \begin{align}
    \traceh \Sch_g^{\parallel}  & \stackrel{{\eqref{traces}}}{=} \frac{\Q}{2 (\n-2)}, \label{r1}\\
    \Sch_g^{\perp}  & \stackrel{{\eqref{traces}}}{=} \frac{1}{\n-1} \left ( \trP - \frac{\epsilon \Q}{2 (\n-2)} \right ), \label{r2}\\
    \cir{\Sch_g^{\parallel}}  & \stackrel{{\eqref{parallelSch}}}{=}
    \frac{1}{\n-2} \left ( \epsilon \circP + \circQ \right ),\label{r3}\\
     \Sch_g^{\perp \parallel} & \stackrel{{\eqref{Schperpar}}}{=} -\frac{1}{\n-2} \trT,  \label{r4}\\
\Weyl_g^{\, \perp} & \stackrel{{\eqref{Weyl4}}}{=} \frac{1}{\n-2}
\left ( (\n-3) \circP - \epsilon \circQ \right )  = \circP - \epsilon \cir{\Sch_g^{\parallel}},\label{r5}\\
     \Weyl_g^{\, \parallel}
         & = \Weyl_u (R^{\parallel}) 
+ \frac{1}{\n -2} \left ( \frac{\circQ}{\n-3} - \epsilon \circP  \right ) 
\owedge \proju \label{r6}\\
\Weyl_g^{\, \perp \parallel} & \stackrel{{\eqref{Weylperpar}}}{=}  \circT.\label{r7}
      \end{align}
  \end{remark}
  We now impose the condition that $\bmu$ is integrable, i.e. $\bmu \wedge d \bmu =0$. Define the acceleration $a = \nabla_{\u} \u$  and the second fundamental form  $\K_{\alpha\beta} := \proju_{\alpha}^{\mu} \proju_{\beta}^{\nu} \nabla_{\mu}  \u_{\nu}$, which is  symmetric  as a consequence of integrability.  The following decomposition holds
  \begin{align*}
    \nabla_{\alpha} \u_{\beta} = \epsilon \u_{\alpha} a_{\beta} + K_{\alpha\beta}.
  \end{align*}
  Recall that for tensor fields $T_{\alpha_1 \cdots \alpha_p}{}^{\beta_1 \cdots \beta_q}$ totally
  orthogonal to $u$ and vector fields $X$ tangential to the leaves (i.e. also
 $u_{\alpha} X^{\alpha}=0$)  the operation  
  \begin{align*}
    D_X T_{\alpha_1 \cdots \alpha_p}{}^{\beta_1 \cdots \beta_q} : =  \proju^{\mu_1}_{\alpha_1} \cdots
      \proju^{\mu_p}_{\alpha_p} \proju^{\beta_1}_{\nu_1} \cdots \proju^{\beta_q}_{\nu_q} \nabla_X T_{\mu_1 \cdots \mu_p}{}^{\nu_1 \cdots \nu_q}
  \end{align*}
  defines a covariant derivative on each leaf that is torsion-free and satisfies $D_X \proju_{\alpha\beta} =0$, $D_X \proju^{\alpha\beta}=0$. Tensors $T_{\alpha_1 \cdots \alpha_p}{}^{\beta_1 \cdots \beta_q}$ of this form restricted to a leaf $\Sigma_{\lambda_0}$
  are in one-to-one correspondence   with
    tensor fields $T_{i_1 \cdots i_p}{}^{j_1 \cdots j_q}$ in $\Sigma_{\lambda_0}$. We use $\proju$ to denote the induced metric (it is associated to $\proju_{\alpha\beta}$ via this map) and $D$
    the corresponding Levi-Civita covariant derivative. The notation is not ambiguous because the tensor $  D_{\hat{X}} T_{i_1 \cdots i_p}{}^{j_! \cdots j_q}$ is associated to $D_X T_{\alpha_1 \cdots \alpha_p}{}^{\beta_1 \cdots \beta_q}$ provided  $X$ is the push-forward of $\hat{X}$. As usual, we let the context specify what is the precise meaning and, in fact we shall not distinguish between $\hat{X}$ and $X$. The tensor field in $M$ totally orthogonal to $u$ corresponding to the Riemann tensor of {$(\Sigma_{\lambda},\proju)$} is
    denoted by $\Rsig_{\alpha\beta\mu\nu}$. In index-free notation we call
  $\Riemsig, \Ricsig, \Scalsig$ respectively  the  tensors in $(M,g)$ associated to
  the Riemann tensor, Ricci tensor and scalar curvature of the metric $\proju$ in the leaves. We also let
  $\Sch_{\Sigma}$ be the tensor in $(M,g)$ associated to the Schouten tensor of
  the leaves, namely
        \begin{align*}
          \Sch_{\Sigma} =
          \frac{1}{\n-3} \left ( \Ricsig
          - \frac{\Scalsig}{2 (\n-2)} \proju \right ) =
          \frac{1}{\n-3} \cir{\Ricsig}
          + \frac{\Scalsig}{2 (\n-1)(\n-2)} \proju
          \end{align*}
          Similarly  $\Weyl_{\Sigma}$ is the tensor in $(M,g)$ associated to the Weyl tensor of $\proju$ in the leaves.

    With these operations introduced, the Gauss, Codazzi and
  Ricci identities are
  \begin{align*}
    \Rtt_{\alpha\beta\mu\nu} & = \Rsig_{\alpha\beta\mu\nu} - \epsilon
    \K_{\alpha\mu} \K_{\beta\nu}+ \epsilon \K_{\alpha\nu} \K_{\beta\mu}, \hfill && \mbox{(Gauss)}
    \\
    \Rnt_{\beta\mu\nu} & = D_{\nu} \K_{\mu\beta}  - D_{\mu} \K_{\nu\beta}, \hfill &&\mbox{(Codazzi)} \\
      \Rnn_{\beta\nu} & = D_{\beta} a_{\nu} - \epsilon a_{\beta}
      a_{\nu} + \K_{\beta\gamma} \K_{\nu}{}^{\gamma} - \pounds_u
      \K_{\beta\nu}. \hfill && \mbox{(Ricci)}
  \end{align*}
  Note that the third identity imposes $D_{\alpha} a_{\beta} = D_{\beta} a_{\alpha}$, which means
  that $a_i$ on each leaf is closed, hence locally exact.
Let us now introduce the trace and trace-free part of
  $\K_{\alpha\beta}$, namely the expansion $\H := \traceh \K$ and
  the shear tensor $\circK_{\beta\nu} = \K_{\beta\nu} - \frac{\H}{\n-1} \proju_{\beta\nu}$. The following result decomposes the Gauss, Codazzi and Ricci identities in terms of the trace and trace-free parts.

  \begin{proposition}
    \label{GCR-decomposed}
    Assume that $(M,g)$ is foliated by injectively immersed hypersurfaces with non-degenerate induced first fundamental form. Using the notation introduced before, the Gauss, Codazzi and Ricci identities are equivalent to
    \begin{align}
      \mbox{Gauss} \quad & \Longleftrightarrow \quad \left \{
                     \begin{array}{l}
  \Q  = \Scalsig + \epsilon \traceh \circK^{(2)}
                                  - \frac{\epsilon(\n-2)}{\n-1} H^2 \\
                 \circQ  = \cir{\Ricsig}
                 + \epsilon \left ( \circK^{(2)} - \frac{\traceh \circK^{(2)}}{\n-1}  \proju \right ) - \frac{\epsilon (\n-3)}{\n-1} \H \circK \\
\Weyl_u(\Rtt) = \Weyl_{\Sigma}- \frac{\epsilon}{2} 
\circK \owedge \circK - \frac{\epsilon}{\n-3} \circK^{(2)} \owedge \proju
                       + \frac{ \epsilon \traceh \circK^{(2)}}{2(\n-2)(\n-3)} \proju \owedge \proju                                \end{array} \right . \label{Gauss}\\
      \nonumber \\
      \mbox{Codazzi} \quad &  \Longleftrightarrow \quad \left \{
                             \begin{array}{l}
                                       \circT_{\beta\mu\nu}  =  D_{\nu} \circK_{\mu\beta}  - D_{\mu} \circK_{\nu\beta} + \frac{1}{\n-2} \left ( \proju_{\beta\mu} D_{\delta} \sigma^{\delta}{}_{\nu}
        - \proju_{\beta\nu} D_{\delta} \sigma^{\delta}{}_{\mu} \right )\\
                               \trT_{\nu}  = \frac{\n-2}{\n-1} D_{\nu} H - D_{\delta} \sigma^{\delta}{}_{\nu}                                             \end{array}
      \right .  \label{Codazzi} \\
      \nonumber\\
      \mbox{Ricci} \quad & \Longleftrightarrow \quad \left \{
                             \begin{array}{l}
                                      \trP = D_{\mu} a^{\mu} - \epsilon a_{\mu} a^\mu - \circK{}_{\mu\nu}
                                      \circK{}^{\mu\nu} - \frac{H^2}{\n-1} - \u(H) \\
                                               \circP_{\beta\nu} =
                                               D_{\beta} a_{\nu} - \epsilon a_{\beta}
                                               a_{\nu} + \frac{1}{\n-1}
                                               \left ( - D_{\mu} a^{\mu}
                                               + \epsilon a_{\mu} a^{\mu}
                                               +  \circK{}_{\rho\sigma} \circK{}^{\rho\sigma} \right ) \proju_{\beta\nu} + \circK_{\beta\gamma} \circK{}_{\nu}{}^{\gamma}
                                                   - \pounds_\u \circK_{\beta\nu}
                                             \end{array}  \right . \label{Ricci} 
    \end{align}
    where $\circK^{(2)}$ is the symmetric $(0,2)$-tensor field
      \begin{align}
        \circK^{(2)}_{\beta\nu} = \circK_{\beta\rho} \circK^{\rho}{}_{\nu}. \label{defsigmasq}
      \end{align}
                     \end{proposition}
  \begin{proof}
  It is immediate that the Codazzi identity can be written as
  \begin{align*}
    & \Rnt_{\beta\mu\nu} =  D_{\nu} \circK_{\mu\beta}  - D_{\mu} \circK_{\nu\beta}
      + \frac{1}{\n-1} \left ( D_\nu H \proju_{\mu\beta} - D_{\mu} H \proju_{\nu \beta} \right )\quad \quad
      \Longleftrightarrow \\
    & \left \{ \begin{array}{ll}
        \circT_{\beta\mu\nu} =  D_{\nu} \circK_{\mu\beta}  - D_{\mu} \circK_{\nu\beta} + \frac{1}{\n-2} \left ( \proju_{\beta\mu} D_{\delta} \sigma^{\delta}{}_{\nu}
        - \proju_{\beta\nu} D_{\delta} \sigma^{\delta}{}_{\mu} \right )\\
               \trT_{\nu} = \frac{\n-2}{\n-1} D_{\nu} H - D_{\delta} \sigma^{\delta}{}_{\nu}
             \end{array}
      \right . 
  \end{align*}
  where the equivalence just follows by taking trace and trace-free parts. This establishes \eqref{Codazzi}.
Concerning the Ricci identity, using
  \begin{align*}
    \pounds_{\u} \u_{\alpha} & = \u^{\mu} \nabla_{\mu} u_{\alpha} - \u_{\alpha}
    \nabla_{\mu} \u^{\alpha} = a_{\mu}, \\
    \pounds_{\u} \proju_{\beta\nu} & = \pounds_{\u} \left ( g_{\beta\nu} {-} \epsilon \u_{\beta} \u_{\nu} \right ) =     \nabla_{\beta} \u_\nu
    +\nabla_{\nu} \u_{\beta} {-} \epsilon \left ( a_{\beta} \u_{\nu}
                                + \u_{\beta} a_{\nu} \right ) = 2 \K_{\beta\nu},
  \end{align*}
it becomes
  \begin{align*}    
         &  \Rnn_{\beta\nu} = D_{\beta} a_{\nu} - \epsilon a_{\beta}
    a_{\nu} + \circK_{\beta\gamma} \circK{}_{\nu}{}^{\gamma}
    - \pounds_\u \circK_{\beta\nu}
    - \frac{1}{\n-1} \left ( \frac{H^2}{\n-1} + \u(H) \right )  \proju_{\beta\nu}
    \quad \quad \Longleftrightarrow  \\
  &      \left \{ \begin{array}{l}
                                      \trP = D_{\mu} a^{\mu} - \epsilon a_{\mu} a^\mu - \circK{}_{\mu\nu}
                                      \circK{}^{\mu\nu} - \frac{H^2}{\n-1} - \u(H)
                                               \quad \quad \quad \mbox{(Raychaudhuri)} \\
                                               \circP_{\beta\nu} =
                                               D_{\beta} a_{\nu} - \epsilon a_{\beta}
                                               a_{\nu} + \frac{1}{\n-1}
                                               \left ( - D_{\mu} a^{\mu}
                                               + \epsilon a_{\mu} a^{\mu}
                                               +  \circK{}_{\rho\sigma} \circK{}^{\rho\sigma} \right ) \proju_{\beta\nu} + \circK_{\beta\gamma} \circK{}_{\nu}{}^{\gamma}
                                                   - \pounds_\u \circK_{\beta\nu}.
                                             \end{array}  \right .
  \end{align*}
  where the equivalence uses $\proju^{\beta\nu} \pounds_{\u} \circK_{\beta\nu} = -
  (\pounds_{\u} \proju^{\beta\nu} ) \circK_{\beta\nu} = 2 K^{\beta\nu} \circK_{\beta\nu}
  = 2 \circK^{\beta\nu} \circK_{\beta\nu}$.  This proves \eqref{Ricci}.
  
  We finally deal with the Gauss identity. In index-free notation it takes the form
  \begin{align*}
    \Rtt = \Riemsig - \frac{\epsilon}{2}  K \owedge K,
  \end{align*}
which upon inserting  $K = \circK + \frac{H}{\n-1} \proju$ becomes
      \begin{align}
        \Rtt = \Riemsig - \frac{\epsilon}{2} \circK \owedge \circK
-  \frac{\epsilon H}{(\n-1)} \circK \owedge \proju
        -  \frac{\epsilon H^2}{2(\n-1)^2} \proju \owedge \proju. \label{Gauss2}
      \end{align}
      Taking into account the following expressions for the traces (recall definition \eqref{defsigmasq})
      \begin{equation}
             \Ric(\circK \owedge \circK) = - 2 \circK^{(2)},\qquad
             \Ric(\circK \owedge \proju) = (\n-3) \circK ,\qquad
             \Ric(\proju \owedge \proju) =  2(\n -2) \proju,
      \end{equation}
      the  $\proju$-trace of \eqref{Gauss2} reads
  \begin{align}
    & \Ric(\Rtt)  = \Ricsig  + \epsilon \circK^{(2)} - \frac{\epsilon (\n-3)}{\n-1} \H \circK
            - \frac{\epsilon (\n-2)}{(\n-1)^2} \H^2 \proju \quad \quad \Longleftrightarrow \nonumber \\
    & \left \{ \begin{array}{l}
                 \Q = \Scalsig + \epsilon \traceh \circK^{(2)}
                                  - \frac{\epsilon(\n-2)}{\n-1} H^2 \\
                 \circQ = \cir{\Ricsig}
                 + \epsilon \left ( \circK^{(2)} - \frac{\traceh \circK^{(2)}}{\n-1}  \proju \right ) - \frac{\epsilon (\n-3)}{\n-1} \H \circK.
                                 \label{trace-Gauss}
                 \end{array} \right .
      \end{align}
        It only remains to extract the trace-free part of the Gauss identity. 
          Directly from the definition \eqref{SecDec} of the map $\Weyl_u$  we get
          \begin{align*}
            \Weyl_u(\Rtt) = \Weyl_{\Sigma}
            - \frac{\epsilon}{2} \left ( K \owedge K 
            - \Sch_u( K \owedge K) \owedge \proju \right ).
          \end{align*}
          We only need to compute
          \begin{align*}
            \circK \owedge \circK
            - \Sch_u ( \circK \owedge \circK) \owedge \proju & =
\circK \owedge \circK + \frac{2}{\n-3} \circK^{(2)} \owedge \proju
- \frac{\traceh \circK^{(2)}}{(\n-2){(\n-3)}} \proju \owedge \proju,
            \end{align*}
            as the pure trace terms (i.e. the last two terms in \eqref{Gauss2}) do not contribute {(cf. Remark \ref{remarkpuretrace})}. We conclude
            \begin{align}
              \Weyl_u(\Rtt) &= \Weyl_{\Sigma}- \frac{\epsilon}{2} 
\circK \owedge \circK - \frac{\epsilon}{\n-3} \circK^{(2)} \owedge \proju
+ \frac{ \epsilon \traceh \circK^{(2)}}{2(\n-2)(\n-3)} \proju \owedge \proju.            \label{ExprW}
        \end{align}
        The  Gauss identity is therefore equivalent to its trace parts
        \eqref{trace-Gauss}  and trace-free part \eqref{ExprW}. This establishes \eqref{Gauss}.
\end{proof}



        \section{Einstein spaces and Fefferman-Graham expansion}\label{secFG}

For a manifold $M$ we denote by $\F(M)$ the set of smooth functions
$f : M \rightarrow \mathbb{R}$ and $\F^\star(M)$ the subset of nowhere vanishing functions. We shall write $\Hess_g$ and $\Delta_g$ respectively for the Hessian and Laplacian of the Levi-Civita connection and $\grad_g f$ for the gradient of $f$. Recall that, in index notation,  $L_{\alpha\beta}$ denotes the
Schouten tensor $\Sch_g$.

In order to analyze the asymptotic properties of our spacetime  we recall the notion of a \emph{conformal compactification} or a \emph{conformal extension}, as originally introduced by Penrose \cite{penroseasymptotics}. Consider an $\n$ dimensional smooth semi-Riemannian (of any signature) manifold $(\tilde M,\tilde g)$ satisfying the Einstein equation (i.e. an Einstein manifold) in the following form
\begin{align}
  \Ric_{\tilde{g}} = (\n -1 ) \lambda \tilde{g}, \qquad \lambda \in \mathbb{R},
  \label{Einst}
\end{align}
or, in terms of the Schouten tensor
\begin{align}
  \Sch_{\tilde{g}} = \frac{\lambda}{2} \tilde{g}, \qquad \lambda \in \mathbb{R}.
  \label{Einstsch}
\end{align}
A smooth manifold $(M,g)$ of the same dimension $\n$ is a conformal extension of $(\tilde M,\tilde g)$ if there exists an embedding $\iota: \tilde M \hookrightarrow M$ such that, under the identification defined by $\iota$, $g = \Om^2 \tilde g$, for some conformal factor $\Om \in \mathcal{F}(M) \cap \mathcal{F}^\star(\tilde M)$, and $M = \tilde M \cup \{ \Om = 0 \}$. The hypersurface $\scri := \{ \Om = 0\}$ represents the asymptotic region of $\tilde g$ and its called \emph{conformal infinity} or \emph{null infinity}. It will be useful sometimes to specify the conformal factor as part of the conformal extension $(M,g;\Om)$.

There are multiple ways of conformally extending a manifold $(\tilde M, \tilde g)$. Namely, given a conformal extension $(M,g;\Om)$, then $(M,\omega^2 g;\omega \Om)$ is also a conformal extension of $(\tilde M, \tilde g)$ for all $\omega \in \mathcal{F}^\star(M)$. In this way, the boundary manifold $\scri$ naturally inherits a conformal class of metrics $[\gamma_0] = \{\omega_0^2 \gamma_0 ~ \slash ~ \omega_0 \in \mathcal{F}^\star(M) \}$ or \emph{conformal structure}. An important result (see e.g. \cite{marspeondata21}, \cite{Grahamlee91} for proof) is the following
\begin{lemma}
 Let $(\tilde M,\tilde g)$ be an Einstein manifold admitting a conformal extension $(M,g;\Om)$. Then for each representative $\gamma_0' \in [\gamma_0]$ of the conformal structure at $\scri$, there exists a conformal extension $(M,g';\Om')$ such that $\grad_{g'} \Om'$ is a geodesic vector in a neighbourhood $U$ of $\{\Om = 0 \}$ and the first fundamental form of induced by $g'$ at $\scri$ coincides with $\gamma'_0$. Moverover, this extension is unique in $U$.
\end{lemma}
\begin{remark}
 We are interested in studying a neighbourhood $U$ of conformal infinity, so we shall often assume that $M=U$ without explicit mention.
\end{remark}

For any conformal extension $(M,g;\Om)$, it is a matter of direct computation to see (e.g. \cite{friedrich02})  that the relation between the Schouten tensors of $\tilde g$ and $g$ is given by
\begin{align}
 \Sch_{\tilde g} - \Sch_g = \frac{1}{\Om} \Hess_g \Om - \frac{1}{2 \Om^2} |\grad_g \Om|^2_g g
\end{align}
which using \eqref{Einstsch}, after rearranging terms, yields
\begin{align}
  \Sch_g = -\frac{1}{\Om} \Hess_g \Om + \frac{1}{2 \Om^2} (|\grad_g \Om|^2_g + \lambda  )g.\label{QEeq0}
\end{align}
\noindent In terms of the scalar
\begin{align*}
  s := \frac{1}{\n} \left ( \Delta_g \Omega + \Omega \tr_g \Sch_g \right )
\end{align*}
the trace of \eqref{QEeq0} gives
 \begin{align}
      |\nabla \Omega|^2_g - 2 s \Omega = - \lambda, \label{deflambda}
    \end{align}
and equation \eqref{QEeq0} can  be rewritten as
\begin{align}
  \Hess_g \Omega + \Omega \, \Sch_g - s g = 0.
  \label{QEeq2}
\end{align}
Whenever $\grad_g \Om$ is a geodesic vector wrt $g$, it follows
\begin{equation}\label{normgrad}
 0 = g^{\alpha \beta} \nabla_\alpha \Om \nabla_\beta \nabla_\mu \Om= g^{\alpha \beta} \nabla_\alpha \Om \nabla_\mu \nabla_\beta \Om= \frac{1}{2} \nabla_\mu (g^{\alpha \beta} \nabla_\alpha \Om  \nabla_\beta \Om) \quad \Longleftrightarrow \quad g^{\alpha\beta} \nabla_\alpha \Om \nabla_\beta \Om = - \lambda,
\end{equation}
where the constant value of $|\grad_g \Om|$ is fixed by \eqref{deflambda} at $\Om = 0$. Note that this also implies that $s = 0$ everywhere, and equation \eqref{QEeq0} for $g$ therefore becomes
\begin{align}
  \Hess_g \Omega + \Omega \, \Sch_g = 0.
  \label{QEeqPoin}
\end{align}
This motivates the following definition:
\begin{definition}
 A triple $(M,g;\Om)$ is a quasi-Einstein manifold if it satisfies \eqref{QEeq2}. If  $\grad_g \Om$ is geodesic with $|\grad_g \Om|_g^2 = -\lambda \neq 0$ then $g$ sastisfies \eqref{QEeqPoin} and  $(M,g;\Om)$ is called a Poincaré quasi-Einstein manifold.
\end{definition}

Note that if $(M,g;\Om)$ is Poincaré quasi-Einstein, the metric decomposes as
\begin{equation}\label{eqpoincareform}
 g = - \frac{d \Om^2}{\lambda} + \proju,
\end{equation}
where $\gamma$ is the projector orthogonal to the unit vector field $u = \grad_g \Om / |\grad_g \Om|_g$. $\gamma$ can be equivalently viewed as a 1-parameter family of $(\n-1)$-metrics defined on the $\{ \Om = const. \}$ leaves. We shall consider next Gaussian coordinates $\{\Omega, x^i \}$ adapted to this foliation. In these coordinates \eqref{QEeqPoin} essentially becomes an equation for $\gamma$, which one can solve order by order in $\Om$, obtaining the so-called Fefferman-Graham expansion \cite{FeffGrah85},\cite{ambientmetric}. The result is  an asymptotic formal series expansion for $\proju$ of the form
\begin{equation}\label{eqexFG}
 \proju \sim \left\lbrace\begin{array}{lr}
              \sum_{k = 0}^{\infty}g_{(2k)} \Om^{2k} + \sum_{l = 1}^{\infty} \obspol_{(l)}(\Om^{\n-1} \log \Om)^l,\quad &\mbox{for $\n$ odd},  \\
              \\
               \sum_{k = 0}^{ k = (\n-2)/2}g_{(2k)} \Om^{2k} + \sum_{k = \n -1}^{\infty}g_{(k)} \Om^{k} ´ ,\quad &\mbox{for $\n$ even}.
             \end{array}
\right.
\end{equation}
The coefficients $g_{(k)}$ depend on $\{ x^i\}$ coordinates and $\obspol_{(l)}$ are smooth family of tensors even in $\Om$, thus they in turn can be Taylor expanded as
\begin{equation}
\obspol_{(l)} \sim \sum_{m = 0}^\infty \obs_{(l,m)} \Om^{2m},
\end{equation}
with coefficients $\obs_{(k,l)}$ also depending on $\{ x^i\}$. The expansions are generated from the derivatives in $\Om$ of \eqref{QEeq2}, which yields a recursive relation of the coefficients when evaluated at $\{\Om = 0 \}$. As it is clear from \eqref{eqexFG}, the results are radically different depending on the parity of $\n$. In rough terms, the idea is as follows:
 \begin{enumerate}
  \item[a)] $\n$ even: Under the assumption of a power series expansion\footnote{This is actually a consequence of smoothness of $g$ at $\scri$.}, the $k-2$ derivative in $\Om$ of the tangential projection of \eqref{QEeq2} (with $s = 0$ because of \eqref{deflambda} together with \eqref{normgrad}) at $\Om = 0$ yields a recursive relation of the form
  \begin{equation}\label{eqrecurFG}
 (\n-1-k) g_{(k)} = \mathcal{F}_{k}(g_{(0)},\cdots,g_{(k-2)}),
\end{equation}
where $\{ \mathcal{F}_{k}(g_{(0)},\cdots,g_{(k-2)})\}$ are functions depending on previous coefficients up to $k-2$. Such relation implies that the prescription of $g_{(0)}$ (which equals the boundary metric $\gamma_0 = \gamma\mid_{\Om = 0}$) generates only non-zero even order terms $g_{(k)}$ with $k < \n - 1$. For $k = \n -1$ the LHS of \eqref{eqrecurFG} vanishes, but so it does $\mathcal{F}_{\n-1}$ identically. This means that, first,  the assumption of a power series expansion is consistent at this order, and second, that the coefficient $g_{(\n-1)}$ is undetermined. After the prescription of $g_{(\n-1)}$, the smooth expansion can be continued consistently to infinity with even and odd order power terms.
  \item[b)] $\n$ odd: The assumption of a power series expansion yields an expression analogous to \eqref{eqrecurFG} also as the $k-2$ derivative in $\Om$ of the tangential projection of \eqref{QEeq2}, where we consider first $k < \n-1$. Thus, similar to the $\n$ even case, the prescription of $g_{(0)}$ generates univocally only non-zero even order terms $g_{(k)}$ with $k < \n-1$. For $k = \n-1$ the LHS of \eqref{eqrecurFG} obviously vanishes, while now $\mathcal{F}_{\n-1}$ is no longer identically zero. This term is known as {\it obstruction tensor} $\obs$, because unless it vanishes, a smooth expansion is obstructed from this order onwards. By the recursive relations of $g_{(k)}$, $\obs$ is a tensor completely determined by $\gamma_0$, so its vanishing restricts the allowed geometries at $\scri$.
  However, by relaxing the smoothness condition, one can keep expanding without imposing $\obs = 0$ if and only if one introduces logarithmic terms in the expansion. Equation \eqref{eqrecurFG} is then modified to  also include dependence on the  coefficients $\obs_{(l,m)}$. In the $\n-1$ order, the logarithmic coefficient appearing is tailored to cancel the obstruction tensor, i.e.
  \begin{equation}
    \obs :=\obs_{(1,0)}.
  \end{equation}
The coefficient $g_{(\n-1)}$ is again not determined because of the multiplicative factor $(\n-1-k)$. Thus, $g_{(\n-1)}$ can be freely prescribed up to some constraints (see below).
  Once the order $k = \n-1$ has been fixed (with prescription of $g_{(\n-1)}$), the resulting expressions keep generating intertwined logarithmic and even power terms to infinite order.

 \end{enumerate}

In summary, the degrees of freedom of the expansion are encoded in the coefficients $(g_{(0)},g_{(\n-1)})$, where $g_{(0)}$ is prescribed by a choice of boundary metric $g_{(0)} = \projuscr=\proju\mid_{\Om = 0}$ and $g_{(\n-1)}$ is also a prescribable symmetric tensor\footnote{This was previously noted by Starobinsky in four spacetime dimensions, positive cosmological constant and Lorentzian signature \cite{starobinsky}.}.
Additionally, the normal-normal and normal-tangent components of \eqref{QEeq2} impose constraints on the trace and divergence of $g_{(\n-1)}$
\begin{equation}\label{eqcontrdiv}
 \tr_\projuscr \gn = \at,\quad\quad\mathrm{div}_\projuscr \gn = \bt,
 \end{equation}
where both $\at$ and $\bt$ are expression involving only $\projuscr$ and its tangential derivatives,  both identically zero if $\n$ is even.

%

A remarkable fact from the formal series analysis is that it actually points out what is the asymptotic data that one can freely prescribe at $\scri$ whenever an asymptotic initial value problem is well-posed. Namely, if $\gamma_0$ is Riemannian and $\lambda >0$, one can pose an initial value problem with data at $\scri$, which are given by $(g_{(0)},g_{(\n-1)})$. Such asymptotic initial value problem happens to be well-posed in the sense that it there is a one-to-one correspondence between the data and Poincaré quasi-Einstein manifolds, as well as continuous  dependence on the asymptotic data (wrt suitable Sobolev norms). The theorems proving these results differ significantly depending on the case. Therefore, rather than providing details on the various approaches, we give a list of references for the interested reader.
For the $\n$ even case one can look at \cite{Anderson2005}, \cite{andersonchrusciel05},\cite{kaminski21}, which reformulate and extend the celebrated results by Friedrich \cite{Fried86initvalue} in four dimensions. The case of arbitrary $\hat{n} \geq 4$  was solved recently by Hintz in \cite{hintz23} (see also \cite{rodniaksi18} for related existence results in the context of the so-called ``ambient space" of Fefferman and Graham.) We also note that, prior to this general result, the $\n$ odd case had been only known under analyticity assumptions \cite{kichenassamy03}.

Consequently, we can give the following characterization result that will be useful in the sequel.

\begin{theorem}\label{theoexisuniq}
  Let $(\Sigma,\gamma_0)$ be an $(\n-1)$-Riemannian manifold equipped with a symmetric covariant 2-tensor $D$ with trace and divergence satisfying the constraints \eqref{eqcontrdiv} and let $\lambda$ be a positive constant. Then, there exists a Poincaré quasi-Einstein manifold $(\Sigma\times [0,\ell),g,\Om)$, for some $\ell \in \mathbb{R}$, with the precribed value of $\lambda$ and $\gamma$ Riemannian,  which admits an expansion \eqref{eqexFG} with seed data $(g_{(0)} = \gamma_{0},g_{(\n-1)} = D)$. Additionally, in the manifold $\tilde M := \Sigma\times (0,\ell)$, the metric $\tilde g := \Om^{-2} g$ satisfies the Einstein equation \eqref{Einst} with the prescribed value of $\lambda$.
\end{theorem}


\subsection{Fall-off of the Weyl tensor components}

We now use the results of Section \ref{secalgebraic} to show how the components of the Weyl tensor decay at $\scri$ for Poincaré quasi-Einstein manifolds. The splitting of the components of the Weyl tensor, as defined in Section \ref{secalgebraic}, will be referred to the unit (geodesic) vector $u := |\lambda|^{-1/2} \grad_g \Om$. As a consequence of $u$ being geodesic, the acceleration vector along its integral curves vanishes
 $a = \nabla_u u = 0.$ Additionally, the second fundamental form satisfies
\begin{equation*}
 \Hess_g \Omega = |\lambda|^{1/2} \K,
\end{equation*}
so that \eqref{QEeqPoin} is
\begin{equation}\label{quasiein0}
  |\lambda|^{1/2} K + \Om \Sch_g = 0
 \end{equation}
 and its traceless part yields
\begin{equation}\label{quasieintf}
 \circK = -\Om |\lambda|^{-1/2} \cir{\Sch_g}.
\end{equation}
Equations \eqref{quasiein0} and \eqref{quasieintf} entail the fall-off of the following objects
\begin{equation}\label{foffcircK}
  \circK = O(\Om),\qquad \K = O(\Om),\qquad H = O(\Om).
\end{equation}
Finally, we will be using Gaussian coordinates $\{\Om, x^i \}$, in which we have
\begin{equation}\label{eqKder}
 u = \epsilon |\lambda|^{1/2} \partial_\Om \qquad \K =\frac{1}{2} \mathcal{L}_u \gamma = \frac{\epsilon}{2} |\lambda|^{1/2}  \partial_\Om \gamma.
\end{equation}

We start by analyzing the fall-off of the components $\Weyl_g^{\, \perp}$ of the Weyl tensor. For that it is useful to derive a general formula using the decomposition in Section \ref{secalgebraic}. A similar expression can be found  in \cite{marspeondata21}, but it is worth rederiving it using the language of Section \ref{secalgebraic}.

\begin{lemma}\label{lemmaformulaweylperp}Let $(M,g;\Om)$ be a Poincaré quasi-Einstein manifold.  Then
\begin{equation}\label{weylperpK}
 \Weyl_g^{\, \perp}  =  K^{(2)}-\pounds_u K +\frac{\epsilon}{\Om}|\lambda|^{1/2} K
\end{equation}
\end{lemma}

\begin{proof}
 Since $K(u,\cdot) = K(\cdot,u) = 0$, from \eqref{quasiein0} we have $\Sch_g = \Sch_g^\parallel$. Therefore, \eqref{quasieintf} yields
\begin{equation}\label{quasiein1}
  |\lambda|^{1/2} \circK + \Om \cir{\Sch_g^\parallel} = 0 .
 \end{equation}
 On the other hand, by \eqref{r5} we have
 \begin{equation}
\Weyl_g^{\, \perp}  = \circP - \epsilon \cir{\Sch_g^{\parallel}}.
 \end{equation}
 Using the second identity  in \eqref{Ricci} (with $a^\nu = \u^\mu \nabla_\mu \u^\nu = 0$) for $\circP$ and \eqref{quasiein1} for  $\cir{\Sch_g^\parallel}$ we obtain
 \begin{equation}\label{wperp1}
  \Weyl_g^{\, \perp}  = \frac{|\circK|^2}{\n-1} \gamma + \circK^{(2)} - \pounds_\u \circK +  \frac{\epsilon}{\Om}|\lambda|^{1/2} \circK,
 \end{equation}
 where $|\circK|^2 =  \circK^{\mu\nu}\circK_{\mu\nu}$.

 The last part of the lemma simply involves expressing $\circK$ as $\circK = K -\frac{\H}{\n-1}\gamma$ in order to show that \eqref{wperp1} can be written as  \eqref{weylperpK}. After simple algebra, \eqref{wperp1} reads
 \begin{align}
  \Weyl_g^{\, \perp}  & = K^{(2)}-\pounds_u K +\frac{\epsilon}{\Om}|\lambda|^{1/2} K + \frac{1}{\n-1}\left( -\frac{\epsilon}{\Om}|\lambda|^{1/2} \H + \u(\H) +|K|^2 \right) \gamma
 \end{align}
 The last term is simplified by first observing the vanishing of the normal-normal component of \eqref{quasiein0} which, taking into account \eqref{r2}, yields
 \begin{equation}
  \Om \Sch_g^{\, \perp} = \frac{\Om}{\n-1} \left ( \trP - \frac{\epsilon \Q}{2 (\n-2)} \right ) = 0 \quad \Longrightarrow \quad  \trP = \frac{\epsilon \Q}{2 (\n-2)} \stackrel{\eqref{r1}}{=}  \epsilon \traceh \Sch_g^{\parallel}.
 \end{equation}
Then, the $\gamma$-trace of \eqref{quasiein0} is
 \begin{equation}
  |\lambda|^{1/2} \H + \Om  \traceh \Sch_g^{\parallel}  =  |\lambda|^{1/2} \H + \Om \epsilon   \trP = |\lambda|^{1/2} \H - \Om \epsilon (|K|^2 + \u(\H)) = 0
 \end{equation}
where for the second equality we have made use of the first identity in \eqref{Ricci} with $|\circK|^2 = |K|^2 -\frac{\H^2}{\n-1}$. Hence
\begin{equation}
 \Weyl_g^{\, \perp}  =  K^{(2)}-\pounds_u K +\frac{\epsilon}{\Om}|\lambda|^{1/2} K .
\end{equation}
\end{proof}

We now use Lemma \ref{lemmaformulaweylperp} to analyze the leading and subleading fall-off terms of $\Weyl_g^{\, \perp}$ that will be useful in Section \ref{secalgspec}. In order to do that all at once it is useful to write the lowest orders of the Fefferman-Graham expansion of $\gamma$ in the following form, unified for all $\n \geq 4$,
\begin{equation}\label{eqexpalln}
 \gamma = \gamma_0 +\Om^2 g_{(2)} + \Omega^3 g_{(3)} \delta_{\n,4} + \Omega^4 g_{(4)} +  \obs \delta_{\n,5}  (\Om^4 \log \Om ) + O(\Om^5),
\end{equation}
whose terms are explained as follows. The lowest order $\gamma_0$ and $g_{(2)}$ appear in all dimensions. Moreover, it always holds that (\cite{ambientmetric} equation (3.6))
\begin{equation}\label{eqg2}
 g_{(2)} = \frac{1}{\lambda} \Sch_{\Sigma_0} = -\frac{\epsilon}{|\lambda|}\Sch_{\Sigma_0}.
\end{equation}
The term $g_{(3)}$ is absent unless $\n = 4$. For all dimensions $\n \geq 4 $ with $\n \neq 5$, $g_{(4)}$ is explicitly computable in terms of $\gamma_0$ and it reads (\cite{ambientmetric} equation (3.18))
 \begin{align}\label{expresiong4}
 & g_{(4)}  =
 \frac{1}{4\lambda^2}\left(-\frac{\Bac_\Sigma}{\n-5} +\Sch_\Sigma^{(2)}  \right),
\end{align}
where $\Bac_{\Sigma_0}$ is the Bach tensor of $\gamma_0$
\begin{equation}
 (\Bac_\Sigma)_{\alpha\beta} := D^\mu (\Cot_\Sigma)_{\alpha\beta\mu} + (\Weyl_\Sigma)_{\alpha\mu\beta\nu}(\Sch_\Sigma)^{\mu\nu}.
\end{equation}
We note that the expressions in \cite{ambientmetric} are derived for the expansion of the so-called ambient metrics, but they generalize to Poincaré-Einstein metrics with any $\lambda \neq 0$ as in \eqref{expresiong4}. Namely, the asymptotic expansions of Lorentzian ambient metrics in \cite{ambientmetric}, in terms of a variable $\rho$, correspond directly to expansions in the variable $\Omega = \sqrt{2\rho}$ of Lorentzian Poincaré-Einstein metrics with $\lambda = 1$. Via the identification $g''_{ij} \rightarrow 8 g_{(4)}$, the formulas in \cite{ambientmetric} remain valid in the latter setting. This correspondence extends to arbitrary positive $\lambda$ via a constant rescaling of the conformal factor, which yields expression \eqref{expresiong4}. The extension to $\lambda < 0$ and Lorentzian signature follows from \cite{anderson}, which shows that the even-order coefficients $g_{(2k)}$, determined entirely by $\gamma_0$, transform between the $\lambda > 0$ (denoted $g_{(2k)}^{dS}$) and $\lambda < 0$ (denoted $g_{(2k)}^{AdS}$) cases via $
g_{(2k)}^{dS} = (-1)^k g_{(2k)}^{AdS}$.

For $\n = 5$ \eqref{expresiong4} no longer holds because $g_{(4)}$ generaly containts a freely specifiable trace-less and transverse component. Additionally, the logarithmic term $\obs \delta_{\n,5}  (\Om^4 \log \Om ) $ only appears for $\n= 5$. Finally, we can always ensure that the tail terms begin at least at order $O(\Omega^5)$, because it either starts with a power term $\Om^5 g_{(5)}$ (if $\n = 4$ or $\n=6$), another logarithmic term $\Om^6\log \Om$ if ($\n = 5$) or higher order terms if $\n >6$.

Thus, combining \eqref{eqKder} and \eqref{eqexpalln} we have
\begin{align}
 K & = \frac{1}{2} \epsilon |\lambda|^{1/2} \partial_\Om \gamma  =\frac{1}{2} \epsilon |\lambda|^{1/2} \left( 2 \Om g_{(2)} + 3 \Omega^2 g_{(3)}\delta_{\n,4} + 4 \Omega^3 g_{(4)} + \obs \delta_{\n,5}  \Om^3(4 \log \Om + 1 )  \right) + O(\Om^4), \label{eqKalln}
\end{align}
so we are now ready to show:
\begin{corollary}\label{lemmanoconflat}
 Let $g$ be a Poincaré quasi-Einstein metric of dimension $\n \geq 4$. Then
 \begin{equation}
  \Weyl_g^{\, \perp} = -\frac{3 |\lambda|}{2}\Om g_{(3)}\delta_{\n,4} + \Om^2 \left( |\lambda|^{-1} \Sch_{\Sigma_0}^{(2) } - 4 |\lambda|g_{(4)} \right) - \obs \delta_{\n,5}  \Om^2(4 \log \Om + 3 ) + O(\Om^3).
 \end{equation}
\end{corollary}
\begin{proof}
 On the one hand, from \eqref{eqg2} and \eqref{eqKalln},
 \begin{align}
  \K^{(2)} & = |\lambda|^{-1} \Sch_{\Sigma_0}^{(2)} \Om^2 + O(\Om^3),\\
  \mathcal{L}_u \K & = \frac{|\lambda|}{2}
  \left( \frac{-2\epsilon}{|\lambda| }\Sch_{\Sigma_0} + 6 \Om g_{(3)}\delta_{\n,4} + 12 \Om^2 g_{(4)} + \obs \delta_{\n,5}\Om^2 (12 \log \Om + 7) \right) + O(\Om^3),\\
  \frac{\epsilon}{\Om}|\lambda|^{1/2} K & = \frac{|\lambda|}{2}
   \left(\frac{-2\epsilon}{|\lambda| }\Sch_{\Sigma_0} + 3 \Omega g_{(3)}\delta_{\n,4} + 4 \Omega^2 g_{(4)} + \obs \delta_{\n,5}  \Om^2(4 \log \Om + 1 )  \right) + O(\Om^3),
 \end{align}
 Then the Corollary follows by substituting the above expresions in Lemma \ref{lemmaformulaweylperp}.
 \end{proof}
\begin{remark}\label{remark45alln}
 This corollary entails some  differences in the leading fall-off term of $\Weyl_g^{\, \perp}$ depending on the dimension: if $\n = 4$, $g_{(3)}$ is non-zero generically so $\Weyl_g^{\, \perp} = O(\Om)$; if $\n = 5$, $g_{(3)} = 0$, but the lowest order term is $\Weyl_g^{\, \perp}  = O(\Om^2 \log \Om) = o(\Om)$; in all dimensions $\n>5$, in general $\Weyl_g^{\, \perp}  = O(\Om^2)$.

 Also note that in all dimensions $\n \neq 5$, we can substitute $g_{(4)}$ according to expression  \eqref{expresiong4}, so that
  \begin{equation}
  \Weyl_g^{\, \perp} = -\frac{3 |\lambda|}{2}\Om g_{(3)}\delta_{\n,4} + \Om^2  \frac{1}{|\lambda|(\n-5)} \Bac_\Sigma + O(\Om^3)\qquad \mbox{($\n \neq 5$)}
 \end{equation}
\end{remark}


The above analysis of $\Weyl_g^\perp$ is particularly interesting in the $\n = 4$ case. Let us denote the leading order term
\begin{equation}\label{ewg3wperp}
 \Wperscr := \Om^{-1}\Weyl_g^\perp\mid_{\scri} =-\frac{3 |\lambda|}{2} g_{(3)}
\end{equation}
This relates a coordinate and conformal gauge dependent quantity $g_{(3)}$ with a tensor $\Om^{-1}\Wperscr$. This can be used, in the four dimensional case, to provide a coordinate and conformal gauge-free definition of the asymptotic data\footnote{Note that such a gauge-free definition follows also from classical results by H. Friedrich \cite{friedrich81},\cite{friedrich81bis},\cite{Fried86initvalue}.} of $\lambda >0$-vacuum metrics, thus providing a geometric characterization of such metrics in a neighbourhood of $\scri$.
 This is a remarkable advantage of the $\n = 4$ case, while for $\n \geq 5$ one generally needs to carry out a Fefferman-Graham expansion to determine $g_{(\n-1)}$.

 Since $\Wperscr$ will be relevant for a later analysis, is convenient to write
 \begin{equation}\label{eqexpWper}
  \Weyl_g^{\, \perp} = \Om \Wperscr \delta_{\n,4} + \Om^2 \left( |\lambda|^{-1} \Sch_{\Sigma_0}^{(2) } - 4 |\lambda|g_{(4)} \right) - \obs \delta_{\n,5}  \Om^2(4 \log \Om + 3 ) + O(\Om^3)
 \end{equation}
and
\begin{equation}
 \K = - |\lambda|^{-1/2} \Om \Sch_{\Sigma_0} - \epsilon|\lambda|^{-1/2} \Omega^2 \Wperscr \delta_{\n,4} + 2\epsilon |\lambda|^{1/2} \Omega^3 g_{(4)} + \frac{\epsilon}{2}|\lambda|^{1/2}\obs \delta_{\n,5}  \Om^3(4 \log \Om + 1 )  + O(\Om^4).\label{eqKalln2}
\end{equation}

We now continue the asymptotic analysis of $\Weyl_g$  with the $ \Weyl_g{}^{\,\perp \parallel}$ components. For that we first note
\begin{equation}
  \Weyl_g{}^{\,\perp \parallel} = \Riem_g{}^{\,\perp \parallel} - \left(\Sch_g \owedge g\right){}^{\,\perp \parallel}.
\end{equation}
For quasi-Einstein Poincaré metrics, it follows from \eqref{quasiein0} and \eqref{perparaA} that $ \left(\Sch_g \owedge g\right){}^{\,\perp \parallel}$ vanishes. Thus, applying the Codazzi equation and \eqref{eqKalln2} gives
\begin{align}
  (\Weyl_g^{\,\perp \parallel})_{\alpha\mu \nu} & = (\Riem_g^{\,\perp \parallel})_{\alpha\mu\nu} = D_\nu \K_{\mu\alpha} - D_\mu \K_{\nu \alpha} = -\Om|\lambda|^{-1/2} (\Cot_{\Sigma_0})_{\alpha \mu \nu} - 2 \epsilon \Om^2 |\lambda|^{-1/2} \delta_{\n,4} D_{[\nu}\Wperscr_{\mu]\alpha} + o(\Om^2),
\end{align}
where
\begin{equation}
 (\Cot_{\Sigma_0})_{\alpha \mu \nu} : = 2 D_{[\nu}(\Sch_{\Sigma_0})_{\mu]\alpha}
\end{equation}
is the Cotton tensor of $\gamma_0$. Note that the tail terms are a Landau little-o because of the presence of logarithmic terms in five dimensions, but otherwise they are $O(\Om^3)$, so for $\n \neq 5$ we may write
\begin{align}\label{eqexpWperpar}
  (\Weyl_g^{\,\perp \parallel})_{\alpha\mu \nu} = -\Om|\lambda|^{-1/2} (\Cot_{\Sigma_0})_{\alpha \mu \nu} - 2 \epsilon\Om^2 |\lambda|^{-1/2} \delta_{\n,4} D_{[\nu}\Wperscr_{\mu]\alpha} + O(\Om^3)\qquad(\n\neq 5.)
\end{align}

\bigskip

We finish our analysis of the fall-off of the Weyl tensor with the $\Weyl_g^{\parallel}$ components. Combining \eqref{r5}, \eqref{r6} and the third equation in identity \eqref{Gauss}, we can write $\Weyl_g^{\, \parallel}$ as
\begin{align}
 \Weyl_g^{\, \parallel}  & = \Weyl_\u(\Rtt) - \frac{\epsilon}{\n-3} \Weyl_g^{\, \perp} \owedge \gamma\nonumber \\
 & = \Weyl_{\Sigma}- \frac{\epsilon}{2}
\circK \owedge \circK - \frac{\epsilon}{\n-3} \circK^{(2)} \owedge \proju
                       + \frac{ \epsilon \traceh \circK^{(2)}}{2(\n-2)(\n-3)} \proju \owedge \proju - \frac{\epsilon}{\n-3} \Weyl_g^{\, \perp} \owedge \gamma. \label{aux1}
\end{align}
In four dimensions this implies that $\Weyl_g^{\,\parallel}$ is  determined by $\Weyl_g^{\, \perp}$ and $\gamma$. Indeed $\Weyl_\Sigma $ is identically zero (because the dimensions of $\Sigma$ is $3$) and so it is
\begin{equation}
 0 = \frac{\epsilon}{2}
\circK \owedge \circK - {\epsilon} \circK^{(2)} \owedge \proju
                       + \frac{ \epsilon \traceh \circK^{(2)}}{4} \proju \owedge \proju,
\end{equation}
because it is a traceless tensor with the symmetries of the Riemann tensor.
Thus
\begin{align}\label{eqweylpar4d}
 \Weyl_g^{\, \parallel}  & =  - \epsilon \Weyl_g^{\, \perp} \owedge \gamma \qquad (\n = 4.)
\end{align}
This equation and the above analysis of $\Weyl_g^{\,\perp}$ and $\gamma$ determine the leading and subleading terms of $\Weyl_g^{\,\parallel}$.
In higher dimensions the leading terms are straightforward to obtain, while the subleading
terms require a more detailed analysis. Since the latter are not needed for the purposes of this paper, we limit our analysis to the leading-order behavior of $\Weyl_g^{\, \parallel}$ for $\n \geq 5$.
Taking into account \eqref{foffcircK} it follows from \eqref{aux1} that
\begin{align}
 \Weyl_g^{\, \parallel} = \Weyl_{\Sigma}- \frac{\epsilon}{\n-3} \Weyl_g^{\, \perp} \owedge \gamma + O(\Om^2). \label{aux3}
\end{align}
A difference now appears between $\n=5$ and all $n>5$ cases. By Remark \ref{remark45alln} it follows
\begin{align}\label{Wparngeq5}
 \Weyl_g^{\, \parallel}   = \Weyl_{\Sigma} + o(\Om)\quad (\n = 5), \qquad \Weyl_g^{\, \parallel}   = \Weyl_{\Sigma} + O(\Om^2)\quad (\n > 5).
\end{align}

\section{Asymptotic structure of algebraically special metrics}\label{secalgspec}

In this section we work out an application of the results obtained so far. We begin with an introduction of the general algebraic classification of the Weyl tensor. Then, in subsection \ref{sec4d} we study how an algebraic condition on the Weyl tensor constraints the asymptotic data of certain $\lambda$-vacuum metrics in four dimensions, leading, in the locally conformally flat $\scri$ case, to a known classification of spactimes related to Kerr-de Sitter \cite{marspaetzseno16}.


The algebraic classification of the Weyl tensor by Coley {\it et al.} \cite{coley04} (see also  \cite{milson}) extends the classical Petrov classification to arbitrary $\n$ dimensions. This  holds for all Weyl-like tensors (denoted $\WT(V)$ in Section \ref{secalgebraic}) of Lorentzian vector spaces. In this section, therefore, we restrict to spacetimes of Lorentzian signature. Additionally, since we will carry out most of our calculations using indices, it is convenient to simplify our notation as follows
\begin{align}\label{eqweylindex}
 \Wi_{\alpha\beta\mu\nu} := (\Weyl_g)_{\alpha\beta\mu\nu},\quad\Wpar_{\alpha\beta\mu\nu} := (\Weyl^{\, \parallel}_g)_{\alpha\beta\mu\nu},\quad\Wppar_{\alpha\mu\nu} := (\Weyl_g^{\,\perp\parallel})_{\alpha\mu\nu},\quad\Wper_{\alpha\beta} := (\Weyl_g^{\,\perp})_{\alpha\beta}.
\end{align}
The above mentioned classification distinguishes a number of cases, called algebraic types, based on the alignment properties of the so-called principal Weyl Align Null Direction (WAND), which are defined as follows.
Let $\wand$ be a null vector field and consider a completion to a semi-null frame $\{\wand,\snlight,\{\snspace_i\}_{i=1}^{\n-2} \}$:
\begin{align}\label{snframe}
 \wand_\alpha \wand^\alpha = \snlight_\alpha \snlight^\alpha = \wand_\alpha (\snspace_i)^\alpha = \snlight_\alpha (\snspace_i)^\alpha= 0,\qquad (\snspace_i)_\alpha (\snspace_j)^\alpha = \delta_{ij},\qquad \wand_\alpha \snlight^\alpha = -1.
\end{align}
Let us denote $\h$ to the projector onto $\mathrm{span}\{\{\snspace_i\}_{i=1}^{\n-2}\}$. Then $\wand$ is said to be a WAND if there exists a semi-null completion of $\wand$ such that
\begin{equation}\label{eqtypeI}
 \Wi_{\alpha\mu\beta\nu}\wand^\alpha\h^\mu{}_{\mu'}\wand^\beta\h^\nu{}_{\nu'} = 0.
\end{equation}
A metric admitting a WAND is said to be algebraically special and of algebraic type $\I$. A subcase of type $\I$ is the so-called type $\Ia$, which happens if in addition
\begin{equation}\label{eqtypeII}
 \Wi_{\alpha\mu\beta\nu}\wand^\alpha\snlight^\mu\wand^\beta\h^\nu{}_{\nu'} = 0.
\end{equation}
Furthermore, the WAND is said to be multiple (and this corresponds to algebraic type $\II$) if, together with the previous conditions, it holds
\begin{equation}\label{eqtypeIIreal}
  \Wi_{\alpha\mu\beta\nu}\wand^\alpha\h^\mu{}_{\mu'}\h^\beta{}_{\beta'}\h^\nu{}_{\nu'} = 0.
\end{equation}
It is worth to remark that from the traceless property of the Weyl tensor, as well as its symmetries, it is easy to see that in four dimensions, equations \eqref{eqtypeII} and \eqref{eqtypeIIreal} are equivalent, so there is no distinction between type $\Ia$ and type $\II$. This is no longer true in higher dimensions.

The types introduced above admit further particularizations, namely, the types $\mathrm{III}$, $\mathrm{N}$ and $\mathrm{O}$, with their corresponding subtypes. For the purposes of this section, we only need to consider algebraic type $\II$.  We refer the interested reader to the original classification \cite{coley04}, \cite{milson}.

 We have reviewed the above criteria in terms of semi-null frames for later convenience. We remark, however, that there exists an equivalent set of frame-free independent criteria for the algebraic classification of the Weyl tensor \cite{ortaggiobeldeb}. As such, \eqref{eqtypeI}, \eqref{eqtypeII} and \eqref{eqtypeIIreal} only depend on $\wand$ and they must hold for any semi-null completion\footnote{A direct test is to observe that any two semi-null frames sharing $\wand$ must differ by a null rotation of $\wand$ and null rotations preserve equations \eqref{eqtypeI}, \eqref{eqtypeII} and \eqref{eqtypeIIreal}.} of a WAND $\wand$. This will be used next to construct an appropriate semi-null frame.

\bigskip

Let us now consider an $\n\geq 4$ dimensional
smooth Einstein manifold $(\tilde M, \tilde g)$, where $\tilde g$ is smooth, Lorentzian of algebraic type at least $\II$. By the results in \cite{durkeereall}, we know that since $\tilde g$ admits a multiple WAND, it must also admit a geodesic multiple WAND $\tilde k$. Let also $(M,g;\Om)$ be a Poincaré quasi-Einstein manifold conformally extending $(\tilde M, \tilde g)$, and define on $(M \backslash \partial M, g)$ the field of one forms $k_\alpha = \tilde k_\alpha$, with associated vector field $k^\alpha := g^{\alpha \beta}k_\beta$. It is clear that $k$ is a multiple WAND of $g$ in $M \backslash \partial M$ and a simple calculation (see e.g. \cite{marspeonKSKdS21}) shows $k$ is also a geodesic field affinely parametrized.
We next use the geodesic property of $k$ to study its extendability to $ \scri = \partial M$ for $\lambda >0$, for which we rely on the Cauchy problem (cf. Theorem \ref{theoexisuniq}). Unfortunately, this strategy cannot be used for $\lambda<0$ cases. Besides, since in this case $\scri$ is timelike, $\wand$ may approach it tangentially, so extendability of $\wand$ need not be guaranteed from a continuity argument. Thus, for $\lambda <0$ we shall later on assume extendability of $\wand$ as part of our hypotheses.

\begin{lemma}\label{lemmaextendk} Let $k$ be a null geodesic vector field of $\lambda>0$ vacuum spacetime admitting a smooth $\scri$. Assume that there exists $U$ a neighbourhood of $\scri$ such that $k$ is nowhere vanishing on $U \backslash \scri$. Then $k$ smoothly extends to $\scri$ as a non-identically zero vector.
\end{lemma}
\begin{proof}

The argument is local so  we may consider a precompact domain (i.e. a non-empty, open and connected subset with compact closure) of the future (spacelike) boundary. We further assume that the boundary of the domain is smooth and non-empty. For simplicity we still denote this compact manifold as $\scri^+$.
Its interior will be denoted $\mathring \scri^+$. We shall restrict to the manifolds $M = D^-(\scri^+)$ and $\mathring M = D^-(\mathring \scri^+)$, where $D^-$ denotes past domain of dependence. Let $\Sigma$ be a Cauchy hypersurface in $M$ and $\mathring \Sigma$ its interior, which is a Cauchy hypersurface of $\mathring M$.
 By construction $\partial \Sigma =  \partial \scri^+$
 and we may also assume that $\Sigma$ and $\scri^+$ only intersect at this boundary, i.e.  $\mathring \Sigma \cap \mathring \scri^+  =\emptyset$.
Also by construction of $M$, geodesics with initial velocity $\wand$ starting at points $p \in M \backslash \scri^+$ must intersect $\scri^+$ at exactly one point (see e.g. \cite{hakwingellis}). Let $\U$ be defined as the subset of points of $\mathring \scri^+$ obtained by intersection with geodesics starting at $\mathring \Sigma$. The corresponding map $\phi$ that sends points in $\mathring{\Sigma}$ to points in ${\cal U}$ is clearly continuous and, by global hyperbolicity it is also a bijection. By construction $\wand$ extends smoothly to ${\cal U}$. We want to show that $\U$ is open and closed in $\mathring \scri^+$ so that $\U = \mathring \scri^+$.

To show that it is open, we first note that $\phi$ is injective. The Theorem of Invariance of Domain   \cite{brouwer} (see also \cite{cao}) ensures that injections from open domains $V$ of $\mathbb{R}^n$ into $\mathbb{R}^n$ are homeomorphisms, namely, $f(V)$ is open. Thus, suitably restricting to coordinate patches, it follows that $\U$ is open in $\mathring \scri^+$.

Let us now $\mathring \partial \U$ denote the boundary of $\U$ in the topology of $\mathring \scri^+$. If $\mathring \partial \U = \emptyset$, closeness follows trivially. Otherwise, consider a point $r \in \mathring \partial \U$,  which observe,  $r \in \mathring \scri^+$.  Let $\{r_i \} \subset \U$ be a Cauchy sequence converging to $r$ and let $\{p_i \} \subset \mathring \Sigma$ the sequence of starting points of the corresponding geodesics.
By compactness of $\Sigma$, there is a subsequence (still denoted  as $\{ p_i\}$) converging,
so let $p$ be the limit point of $\{ p_i\}$. If $p \in \partial \Sigma$,
 since $\partial \Sigma = \partial \scri^+$, it must be that $r \in \partial \scri^+$.  More precisely $p = r$ because they both lie in the same spacelike surface, and they are, respectively, limit points of sequences $\{p_i\}$ and $\{r_i\}$, whose elements are joined by a null geodesic.
 Therefore $p \in \partial \Sigma$ is in contradiction with the fact that $r \in \mathring \scri^+$.
Hence, it can only be that $p \in \mathring \Sigma$,
 then $r \in \U$ so $\mathring\partial U \subset \U$. Thus, $\U$ is closed in $\mathring\scri^+$.



\end{proof}

We now return to arbitrary sign of $\lambda \neq 0$. We select a multiple WAND $k$ on a neighbourhood of $\scri$, which we assume it extends transversally to $\scri$ if $\lambda<0$. The latter assumption is not necessary if $\lambda >0$ because of Lemma \ref{lemmaextendk}.  As before, let  $u := |\lambda|^{-1/2} \grad_g \Om$ be the unit non-null vector of squared norm $\epsilon$ and let also $\y$ be an orthogonal complement to $\u$ in $\mathrm{span}\{\wand,\u \}$, which has squared norm $-\epsilon$. Define the nowhere zero function $s : = \epsilon u_\alpha \wand^\alpha$ . Then we can decompose $\wand$ and define the null vector $\snlight$ as
\begin{equation}\label{deckl}
 \wand = s( \u + \y),\qquad \snlight := -\frac{\epsilon}{2} s^{-1}( \u -\y),
\end{equation}
 which we complete to a semi-null frame \eqref{snframe}. Thus, using the notation introduced in \eqref{eqweylindex}  and \eqref{deckl}, after suitably removing non-zero common factors, type $\I$ condition \eqref{eqtypeI} reads
\begin{equation}\label{eqtypeI5d}
 \Wper_{\mu\nu}\h^\mu{}_{\mu'}\h^\nu{}_{\nu'} + 2\Wppar_{(\mu|\alpha|\nu)}\y^\alpha\h^\mu{}_{\mu'}\h^\nu{}_{\nu'} + \Wpar_{\alpha\mu\beta\nu}\y^\alpha\y^\beta\h^\mu{}_{\mu'}\h^\nu{}_{\nu'} = 0,
\end{equation}
while type $\Ia$  \eqref{eqtypeII} and type $\II$  \eqref{eqtypeIIreal} become, respectively,
 \begin{align}
 \Wper_{\mu\nu}\y^\mu\h^\nu{}_{\nu'} + \Wppar_{\alpha\beta\nu}\y^\alpha \y^\beta\h^\nu{}_{\nu'} & = 0,\label{eqtypeII5d}\\
  \Wppar_{\mu\beta\nu}\h^\mu{}_{\mu'}\h^\beta{}_{\beta'}\h^\nu{}_{\nu'} + \Wpar_{\alpha\mu\beta\nu}\y^\alpha\h^\mu{}_{\mu'}\h^\beta{}_{\beta'}\h^\nu{}_{\nu'}& = 0. \label{eqtypeII5d2}
\end{align}


\subsection{Four dimensional case}\label{sec4d}

In four dimensions, inserting equation \eqref{eqweylpar4d} into \eqref{eqtypeI5d}, \eqref{eqtypeII5d} and \eqref{eqtypeII5d2} provides a relation between $\Weyl_g^{\, \perp \parallel}$ and $\Weyl_g^{\, \perp}$.
First we expand \eqref{eqweylpar4d} using  \eqref{KN}
\begin{align}
 \Wpar_{\alpha\beta\mu\nu} =-\epsilon(\Weyl_g^\perp \owedge \gamma)_{\alpha\beta\mu\nu} = -\epsilon(\Wper_{\alpha\mu} \gamma_{\beta \nu} -\Wper_{\alpha\nu} \gamma_{ \beta\mu}- \Wper_{\beta\mu} \gamma_{\alpha \nu} + \Wper_{\beta\nu} \gamma_{\alpha \mu})
\end{align}
and insert it into \eqref{eqtypeI5d}, \eqref{eqtypeII5d} and \eqref{eqtypeII5d2} to get
\begin{align}
 2 \Wper_{\mu\nu}\h^\mu{}_{\mu'}\h^\nu{}_{\nu'} + 2\Wppar_{(\mu|\alpha|\nu)}\y^\alpha\h^\mu{}_{\mu'}\h^\nu{}_{\nu'} -\epsilon \Wper_{\alpha\beta}\y^\alpha \y^\beta \h_{\mu'\nu'} & = 0,\label{tI4D}\\
 \Wper_{\mu\nu}\y^\mu\h^\nu{}_{\nu'} + \Wppar_{\alpha\beta\nu}\y^\alpha \y^\beta\h^\nu{}_{\nu'} & = 0,\label{tII4D}\\
  \Wppar_{\mu\beta\nu}\h^\mu{}_{\mu'}\h^\beta{}_{\beta'}\h^\nu{}_{\nu'} -2\epsilon \Wper_{\alpha\mu}\y^\alpha \h^\mu{}_{[\mu'}\h_{\nu']\beta'} & = 0.\label{tII4Db}
\end{align}
Recall that, in four dimensions, equation \eqref{tII4Db} follows from \eqref{tII4D} and the traceless property of the Weyl tensor, so without loss of generality we shall only consider  \eqref{tI4D} and  \eqref{tII4D} in the remainder.

Let us now decompose $\Wper$ into components parallel and perpendicular to $\y$:
\begin{align}
 \Wper_{\mu'\nu'} = \Wper_{\alpha\beta}\y^\alpha \y^\beta \y_{\mu'}y_{\nu'} -2 \epsilon \Wper_{\alpha\nu}\y^\alpha \y_{(\mu'}\h^{\nu}{}_{\nu')} + \Wper_{\mu\nu}\h^\mu{}_{\mu'} \h^\nu{}_{\nu'}.
\end{align}
Substituting $\Wper_{\mu\nu}\h^\mu{}_{\mu'} \h^\nu{}_{\nu'}$ according to \eqref{tI4D} and $\Wper_{\alpha\nu}\y^\alpha \h^{\nu}{}_{\nu'}$ according to \eqref{tII4D} gives
\begin{align}
 \Wper_{\mu'\nu'} & = \Wper_{\alpha\beta}\y^\alpha \y^\beta \y_{\mu'}y_{\nu'}
 +2 \epsilon \Wppar_{\alpha\beta\nu}\y^\alpha \y^\beta \y_{(\mu'}\h^{\nu}{}_{\nu')}
 - \Wppar_{(\mu|\alpha|\nu)}\y^\alpha\h^\mu{}_{\mu'}\h^\nu{}_{\nu'} +\frac{1}{2}\epsilon \Wper_{\alpha\beta}\y^\alpha \y^\beta \h_{\mu'\nu'} \\
 & =
  \Wper_{\alpha\beta}\y^\alpha \y^\beta \left(\y_{\mu'}y_{\nu'}  + \frac{\epsilon}{2} \h_{\mu'\nu'}\right)
 +2 \epsilon \Wppar_{\alpha\beta\nu}\y^\alpha \y^\beta \y_{(\mu'}\h^{\nu}{}_{\nu')}
 - \Wppar_{(\mu|\alpha|\nu)}\y^\alpha\h^\mu{}_{\mu'}\h^\nu{}_{\nu'} \\
 & =
  \frac{3}{2}\Wper_{\alpha\beta}\y^\alpha \y^\beta \left(\y_{\mu'}y_{\nu'}  + \frac{\epsilon}{3} \gamma_{\mu'\nu'}\right)
 +2 \epsilon \Wppar_{\alpha\beta\nu}\y^\alpha \y^\beta \y_{(\mu'}\h^{\nu}{}_{\nu')}
 - \Wppar_{(\mu|\alpha|\nu)}\y^\alpha\h^\mu{}_{\mu'}\h^\nu{}_{\nu'}\label{WpertypeII}
\end{align}
where in the last step we use $\h_{\mu'\nu'} = \gamma_{\mu'\nu'}+\epsilon \y_{\mu'} \y_{\nu'} $. Note that the contraction of \eqref{WpertypeII} with $y^{\mu'} y^{\nu'}$ gives a trivial identity. Thus, it does not determine the term $W^{\perp}_{\alpha\beta} \y^{\alpha}\y^{\beta}$. This observation will be relevant in the following arguments in order to keep track of the degrees of freedom in the final expressions. Namely, analyzing the leading order term in $\Omega$ of \eqref{WpertypeII} we shall find an expression for the electric part of the rescaled Weyl tensor at $\scri$, $\Wperscr = \Om^{-1} \Wper\mid_{\Omega = 0}$, in terms of the metric at $\scri$, $\gamma_0$,  a vector field $\y$ and a free function $\f$ (related to  $W^{\perp}_{\alpha\beta} \y^{\alpha}\y^{\beta}$) at $\scri$. Additionally, from the subleading order in $\Omega$ of \eqref{WpertypeII} we shall also obtain a differential  relation constraining  $\Wperscr,\gamma_0, \y$ and $f$.



\bigskip


When evaluating a quantity at $\scri$ we shall indicate this by using the symbol $``\eqscr"$ in the first equality and restriction to $\scri$ will be implicit thereon.
Using this notation, the electric part of the rescaled Weyl tensor is  (cf. \eqref{eqexpWper})
\begin{equation}
 \fop \Wper \eqscr  \Wperscr.
\end{equation}
Additionally, we need to define the acceleration of the vector $\y$ at $\scri$
\begin{equation}\label{acyscri}
 \y^\alpha D_\alpha \y^\beta \eqscr: \ay^\beta
\end{equation}

To obtain the first order equation in $\Om$ of \eqref{WpertypeII} we apply $\fop$ and evaluate at $\scri$, substituting the derivatives of $\Wper$ and $\Wppar$ by the corresponding terms in \eqref{eqexpWper} and  \eqref{eqexpWperpar} respectively. A similar idea applies to second order in $\Om$ of \eqref{WpertypeII}, although this time we apply $(1/2)\fop^2$ and evaluate at $\scri$. The details of the calculations are omitted here; a complete derivation is provided in Appendix \ref{appexp}.

The first order term in the $\Om$ expansion of \eqref{WpertypeII} gives the equation
\begin{align}
 \Wperscr_{\mu'\nu'}   \eqscr
  \mathcal{F} \left(\y_{\mu'}y_{\nu'}  + \frac{\epsilon}{3} \gamma_{\mu'\nu'}\right)
 -2\epsilon|\lambda|^{-1/2}  (\Cot_{\Sigma})_{\alpha\beta\nu}\y^\alpha \y^\beta \y_{(\mu'}\h^{\nu}{}_{\nu')}
 +|\lambda|^{-1/2} (\Cot_{\Sigma})_{(\mu|\alpha|\nu)}\y^\alpha\h^\mu{}_{\mu'}\h^\nu{}_{\nu'}, \label{WpertypeIIO1}
\end{align}
while the $\Om^2$ order yields
\begin{align}
  (\Bac_\Sigma)_{\mu'\nu'}
  =&  \mathcal{G}\left(\y_{\mu'}\y_{\nu'} + \frac{\epsilon}{3}\gamma_{\mu'\nu'}\right) + 2 \epsilon \F   |\lambda|^{1/2} \ay_{(\mu'}\y_{\nu')}  \\
  +  &   (\Cot_{\Sigma})_{\alpha\beta\nu}\left(\y^{\alpha}\y^{\beta}\ay^\nu \h_{\mu'\nu'} -  4 \ay^{(\alpha}\y^{\beta)}\y_{(\mu'}\h^\nu{}_{\nu')}-2\y^\alpha \y^\beta  \h^{\nu}{}_{(\nu'} \ay_{\mu')} \right) \\
 +& (\Cot_\Sigma)_{\mu\alpha\nu}\left( \epsilon \ay^\alpha h^{\mu}{}_{(\mu'}  h^{\nu}{}_{\nu')} +2  \y^\alpha\y^{(\mu} \h^{\nu)}{}_{(\nu'}\ay_{\mu')} +2 \y^\alpha\ay^{(\mu} \h^{\nu)}{}_{(\nu'}\y_{\mu')}\right)\\
 + & 4 |\lambda|^{1/2}  D_{[\nu}\Wperscr_{\beta]\alpha}\y^\alpha \y^\beta\y_{(\mu'}\h^\nu{}_{\nu')} - 2\epsilon|\lambda|^{1/2}D_{[\mu} \Wperscr_{\alpha]\nu}\y^\alpha\h^\mu{}_{(\mu'}\h^\nu{}_{\nu')} . \label{2ndOfull}
\end{align}
Where $\mathcal{F}$ and $\mathcal{G}$ are scalar functions prescribing, respectively, $\frac{3}{2}\Wperscr_{\alpha\beta}\y^\alpha \y^\beta$ and $\frac{3}{2}(\Bac_\Sigma)_{\alpha\beta}\y^\alpha \y^\beta$.
We interpret the system of equations \eqref{WpertypeIIO1}–\eqref{2ndOfull} as follows. Equation \eqref{WpertypeIIO1} expresses $\Wperscr$ in terms of the asymptotic metric $\gamma$, a vector field $\y$, and a scalar function $\mathcal{F}$. Substituting \eqref{WpertypeIIO1} into \eqref{2ndOfull} transforms the latter into a constraint differential equation for $\mathcal{F}$, $\mathcal{G}$, $\y$, and $\Cot_\Sigma$. Thus, for a fixed choice of $\gamma$, one may attempt to solve \eqref{2ndOfull} to determine $\mathcal{F}$, $\mathcal{G}$ and $\y$, and subsequently determine $\Wperscr$ via \eqref{WpertypeIIO1}.

This is particularly relevant in the case $\lambda > 0$, where $(\gamma, \Wperscr)$ determines the asymptotic initial data for $\lambda > 0$ vacuum spacetimes (see Theorem \ref{theoexisuniq} and \eqref{ewg3wperp}). A solution to the system \eqref{WpertypeIIO1}–\eqref{2ndOfull} therefore defines a class of initial data corresponding to a subset of algebraically special metrics with prescribed asymptotic geometry.
Whether this class exhausts all algebraically special metrics with a given $\gamma$ depends on whether satisfying \eqref{WpertypeIIO1}–\eqref{2ndOfull} asymptotically is a sufficient condition for the spacetime to be of type $\mathrm{II}$ in a neighborhood of $\scri$. This question would require an analysis of the evolution of the Weyl tensor, for instance, through the second Bianchi identity. While it is not implausible that \eqref{WpertypeIIO1}–\eqref{2ndOfull} may turn out to be sufficient, this question lies beyond the scope of this paper.

%

\bigskip

In this paper, we illustrate the idea in the previous paragraph restricting ourselves to the simplest scenario, namely, imposing  $\gamma$ to be locally conformally flat. This means
 that $\Cot_\Sigma = 0$ and \eqref{WpertypeIIO1} reads
\begin{align}
 \Wperscr_{\alpha\beta} =
  \mathcal{F} \left(\y_{\alpha}y_{\beta}  + \frac{\epsilon}{3} \gamma_{\alpha\beta}\right).
\end{align}
Now two possibilities arise. If $\F = 0$ then $\Wperscr = 0$ and the spacetime constructed from the data $(\gamma,\Wperscr = 0)$ is  locally diffeomorphic to de Sitter in a neighbourhood of $\scri$ (see e.g. \cite{Mars_scri1}), which concludes the analysis for this case. When  $\mathcal{F} \neq 0$ it turns out to be useful to  redefine  $\mathcal{F} = \f^{-3}$ so that
\begin{align}
 \Wperscr_{\alpha\beta} =
  \frac{1}{\f^3} \left(\y_{\alpha}y_{\beta}  + \frac{\epsilon}{3} \gamma_{\alpha\beta}\right).\label{WperloconflatO1}
\end{align}
From \eqref{WperloconflatO1} it is clear that in this case
\begin{align}\label{Wperscrayy}
 \Wperscr_{\alpha\beta} \ay^{(\alpha}\y^{\beta)} = 0,
\end{align}
 because recall $\ay^\alpha = \y^\beta D_\beta \y^\alpha$ and therefore $\ay_\alpha \y^\alpha = 0$.

 On the other hand,  $\gamma$ locally conformally flat implies that $\Bac_\Sigma = 0$. Thus, by definition $\mathcal{G} = 0$ (cf. \eqref{defG}), so   \eqref{2ndOfull} becomes,
after factoring out  $|\lambda|^{1/2}$,
 \begin{align}
   0
  = \frac{2 \epsilon }{f^3}  \ay_{(\mu'}\y_{\nu')}
  + & 4   D_{[\nu}\Wperscr_{\beta]\alpha}\y^\alpha \y^\beta \h^{\nu}{}_{(\nu'}\y_{\mu')}
  -  2\epsilon D_{[\mu} \Wperscr_{\alpha]\nu}\y^\alpha\h^\mu{}_{(\mu'}\h^\nu{}_{\nu')} . \label{WpertypeIIconflatO2}
\end{align}
Our claim now is that equation \eqref{WpertypeIIconflatO2} implies that the vector field $\xi := f \y$ is a conformal Killing vector field of $\gamma$. This will allow us to exploit the classification results obtained in \cite{marspeonKSKdS21}, as we describe in more detail below.

 Using  \eqref{WperloconflatO1} the derivative $D_{\nu} \Wperscr_{\beta\alpha}$ is
\begin{align}
  D_{\nu}\Wperscr_{\beta\alpha} = \frac{-3}{f^4} D_\nu f \left( \y_\beta \y_\alpha + \frac{\epsilon}{3} \gamma_{\beta \alpha}\right) + \frac{1}{f^3} \left(\y_\alpha D_\nu \y_\beta  + \y_\beta D_\nu \y_\alpha \right)
\end{align}
from which
\begin{align}
  D_{[\nu}\Wperscr_{\beta]\alpha}\y^\alpha \y^\beta \h^{\nu}{}_{(\nu'}\y_{\mu')}  &= \frac{-3}{f^4} D_{[\nu }f \left( \y_{\beta]} \y_\alpha + \frac{\epsilon}{3} \gamma_{\beta] \alpha}\right) \y^\alpha \y^\beta \h^{\nu}{}_{(\nu'}\y_{\mu')} + \frac{1}{f^3} \left( \y_\alpha D_{[\nu} \y_{\beta]}  + \y_{[\beta} D_{\nu]} \y_\alpha \right) \y^\alpha \y^\beta \h^{\nu}{}_{(\nu'}\y_{\mu')} \\ & =
 \frac{-1}{f^4} D_\nu f  \h^{\nu}{}_{(\nu'}\y_{\mu')} +
  \frac{1}{f^3} \frac{\epsilon}{2} \y^\beta D_{\beta} \y_{\nu}    \h^{\nu}{}_{(\nu'}\y_{\mu')} =
  \frac{-1}{f^4} D_\nu f  \h^{\nu}{}_{(\nu'}\y_{\mu')} +
  \frac{1}{f^3} \frac{\epsilon}{2} \ay_{(\nu'}\y_{\mu')}
\end{align}
and
\begin{align}
   D_{[\mu} \Wperscr_{\alpha]\nu}\y^\alpha\h^\mu{}_{(\mu'}\h^\nu{}_{\nu')}   &= \frac{-3}{f^4} D_{[\mu }f \left( \y_{\alpha]} \y_\nu + \frac{\epsilon}{3} \gamma_{\alpha] \nu}\right) \y^\alpha\h^\mu{}_{(\mu'}\h^\nu{}_{\nu')}  +
   \frac{1}{f^3} \left(\y_\nu D_{[\mu} \y_{\alpha]}  + \y_{[\alpha} D_{\mu]} \y_\nu \right) \y^\alpha\h^\mu{}_{(\mu'}\h^\nu{}_{\nu')}  \\ & =
 \frac{1}{f^4} \frac{\epsilon}{2}\y^\alpha D_\alpha f  \h_{\mu'\nu'}
 - \frac{1}{f^3} \frac{\epsilon}{2} D_{\mu} \y_{\nu}   \h^\mu{}_{(\mu'}\h^\nu{}_{\nu')}.
\end{align}
Thus \eqref{WpertypeIIconflatO2} reads
\begin{align}
   0
  & =
  \frac{2\epsilon}{f^3}   \ay_{(\mu'}\y_{\nu')}
  -
    \frac{4}{f^4} D_\nu f  \h^{\nu}{}_{(\nu'}\y_{\mu')}
    +
  \frac{2\epsilon}{f^3}  \ay_{(\nu'}\y_{\mu')}
  - \frac{1}{f^4} \y^\alpha D_\alpha f  \h_{\mu'\nu'}
 + \frac{1}{f^3} D_{\mu} \y_{\nu}   \h^\mu{}_{(\mu'}\h^\nu{}_{\nu')} \\ & = -\frac{4}{f^3} \left( \frac{1}{f}D_{\nu} f - \epsilon \ay_\nu \right)\h^{\nu}{}_{(\nu'}\y_{\mu')} - \frac{1}{f^4} \y^\alpha D_\alpha f  \h_{\mu'\nu'}
 + \frac{1}{f^3} D_{\mu} \y_{\nu}   \h^\mu{}_{(\mu'}\h^\nu{}_{\nu')}. \label{WpertypeIIconflatO2v2}
\end{align}
Now consider the following standard decomposition of the covariant derivative of a unit vector field $\y_\alpha$ in components parallel and orthogonal to $\y_\alpha$
\begin{equation}\label{eqdecnaby}
 \nabscr_\alpha \y_\beta = -\epsilon \y_\alpha \ay_\beta + \Pi_{\alpha \beta} + \frac{\h_{\alpha \beta}}{2} \L +  \w_{\alpha \beta},\qquad \L:= \nabscr_\alpha \y^\alpha,
\end{equation}
where
\begin{equation}
 \Pi_{(\alpha \beta)} = \Pi_{\alpha \beta},\quad\quad {\Pi^\alpha}_\alpha = 0, \quad\quad \w_{[\alpha \beta]} = \w_{\alpha \beta},
\end{equation}
satisfy
\begin{equation}
 \y^\alpha \Pi_{\alpha \beta} = \y^\alpha \h_{\alpha \beta} = \y^\alpha \w_{\alpha \beta} = 0.
\end{equation}
In terms of these objects,  \eqref{WpertypeIIconflatO2v2} is
\begin{align}
0 =
& -4\left( \frac{1}{\f^4} D_\nu \f - \frac{\epsilon}{\f^3}\ay_\nu \right)\h^\nu{}_{(\nu'}\y_{\mu')}
- \left(\frac{1}{\f^4}\y^\alpha D_\alpha \f - \frac{\L}{2\f^3}\right) \h_{\mu' \nu'}
+\frac{1}{f^3} \Pi_{\mu'\nu'}.
\end{align}
One contraction with $\y^{\mu'}$ gives
\begin{align}\label{CKeq1}
  \h^\nu{}_{\nu'} D_\nu \f - \epsilon\f \ay_{\nu'}  = 0
\end{align}
while the projection orthogonal to $\y$ has, after multiplication by $\f^4$, trace and traceless parts
\begin{align}\label{CKeq2}
 \y^\alpha D_\alpha \f - \frac{\f\L}{2} = 0, \qquad \Pi_{\mu'\nu'} = 0 .
\end{align}
Equations \eqref{CKeq1}-\eqref{CKeq2} are a recast of the conformal Killing equation of $\xi = \f \y$, namely,
\begin{equation}\label{CKeqxi}
 2 D_{(\mu} \xi_{\nu)}  - \frac{2}{3} D_\alpha \xi^\alpha \gamma_{\mu \nu} = 0.
\end{equation}
This follows by subsituting $\xi = f \y$ and writing the result in terms of the quantities in \eqref{eqdecnaby}
\begin{align}
 2 D_{(\mu} \xi_{\nu)}  & = 2 \y_{(\mu} D_{\nu)}f + 2 f  D_{(\mu}y_{\nu)} = 2 \y_{(\mu} D_{\nu)}f  -2 \f \epsilon \y_{(\mu} \ay_{\nu)} + 2\f \Pi_{\mu\nu} +  \f \L \h_{\mu\nu},\\
 \frac{2}{3} D_\alpha \xi^\alpha \gamma_{\mu \nu} & = \frac{2}{3} (\y^\alpha D_\alpha f + f  D_\alpha \y^\alpha   ) \gamma_{\mu \nu} =\frac{2}{3} (\y^\alpha D_\alpha f + f  \L   ) \gamma_{\mu \nu},
\end{align}
so that \eqref{CKeqxi} takes the form
\begin{equation}\label{CKeqxi2}
  2 \y_{(\mu} D_{\nu)}f  -2 \f \epsilon \y_{(\mu} \ay_{\nu)} + 2\f \Pi_{\mu\nu} +  \f \L \h_{\mu\nu}  - \frac{2}{3} (\y^\alpha D_\alpha f + f  \L   ) \gamma_{\mu \nu} = 0.
\end{equation}
A full projection orthogonal to $\y$ of \eqref{CKeqxi2} yields
\begin{align}
  2\f \Pi_{\mu\nu} +  \f \L \h_{\mu\nu} - \frac{2}{3} (\y^\mu D_\mu f + f  \L   )\h_{\mu\nu} = 0,
\end{align}
whose trace and traceless parts are readily
equivalent to \eqref{CKeq2}. Similarly, a projection of \eqref{CKeqxi2} parallel to $\y$ in one index and orthogonal in the other index gives \eqref{CKeq1}. Also note that two contractions with $\y$ recover again the first equation in \eqref{CKeq2}, meaning that \eqref{CKeqxi2} is entirely equivalent to \eqref{CKeq1}-\eqref{CKeq2}.

\bigskip

In summary, we  have examined the asymptotic implications of the algebraic type $\II$ condition of the Weyl tensor up to  second order in $\Om$. In the cases in which the metric induced at $\scri$, $\gamma$, is locally conformally flat, this fixes the leading order of $\Wper_{\alpha\beta}$ to be of the form
\begin{align}
 \Wperscr_{\alpha\beta}
=  \frac{\kappa}{|\xi|_\gamma^5} \left(\xi_{\alpha}\xi_{\beta}  + \epsilon \frac{|\xi|_\gamma^2}{3} \gamma_{\alpha\beta}\right),\qquad \kappa  \in \{ 0,\pm 1\} , \label{WperloconflatO1c}
\end{align}
where $\xi = f \y$ and  $\kappa f := |f| =  |\xi|_{\gamma}$. The value $\kappa = 0$ is included so that \eqref{WperloconflatO1c} covers the case $\Wper_{\alpha\beta} = 0$ (see discussion before equation \eqref{WperloconflatO1}).
Note that \eqref{WperloconflatO1c} is just a recast of \eqref{WperloconflatO1} in terms of $\xi$, which we have  shown to be a conformal Killing vector of $\gamma$. The results up to this point hold for any sign of non-zero $\lambda$. However, as previously discussed, $(\gamma,\Wperscr)$ fixes the asymptotic data in the $\lambda>0$ vacuum case (cf. Theorem \ref{theoexisuniq} and equation \eqref{ewg3wperp}).

The above discussion can be reformulated as a theorem:
\begin{theorem}\label{theoKdSlike}
 Let $(\tilde M,\tilde g)$ be a a four dimensional Einstein manifold  with $\lambda \neq 0$, admitting a locally conformally flat $\scri$, and of algebraic type at least $\II$ with geodesic multiple WAND $\wand$. Assume that $\wand$ extends transversally to $\scri$ if $\lambda <0$.
 Then, the electric part of the Weyl tensor at $\scri$, $\Wperscr$, is of the form \eqref{WperloconflatO1c}, where $\xi$ is a conformal Killing vector of $\gamma$.
\end{theorem}
Restricting now to the case $\lambda > 0$, the class of asymptotic data described in Theorem \ref{theoKdSlike} has previously been studied in \cite{marspaetzseno16}. In that work, such data at $\scri$ arise as a consequence of imposing that the bulk spacetime satisfies the so-called alignment property between the Weyl tensor\footnote{See \cite{Mars_Kerr}, \cite{MarsSenovilla} for a precise definition of this property and its relation to the Kerr and Kerr-de Sitter solutions.} and a Killing vector field, together with the assumption that $\scri$ is locally conformally flat. Under these conditions, all metrics are explicitly constructed and classified, and are referred to as the Kerr-de Sitter-like class of metrics (see also \cite{Mars_scri1} for a related classification without imposing local conformal flatness of $\scri$.) As the name suggests, this class includes the Kerr-de Sitter family, but also contains additional families of solutions.

Extension of the definition to higher dimensions and further properties of this class (both in four and higher dimensions)
is provided in \cite{marspeonKSKdS21}, where
the Kerr-de Sitter like-class is shown to be characterized as (i.e. to be fully equivalent to)  the class of $\lambda>0$-vacuum Kerr-Shild metrics constructed from a de Sitter background and sharing the same null infinity as its background\footnote{A recent result \cite{chruscielfinnian}  constructs metrics with $\lambda<0$ by analytic extension of Kerr-Anti-de Sitter, thus admitting a Kerr-Schild form and sharing $\scri$ with its background. The method reminds the one employed in \cite{marspeonKSKdS21} to construct the so-called Wick-Rotated-Kerr de Sitter family, so potential connections may be established in the future.}. This hints a potential connection between the Kerr-de Sitter-like class and the set of all algebraically special metrics admitting a conformally flat $\scri$, because Kerr-Schild metrics are known \cite{malekpravda11} to be of algebraic type at least $\II$. Theorem \ref{theoKdSlike} makes this connection precise, thereby providing another interesting characterization of the Kerr-de Sitter-like class. A similar characterization may be developed in higher dimensions. However, the analysis in this case is more involved and will be addressed in a future work.

\appendix

\section{Second order expansion of algebraic type $\II$ equations}\label{appexp}


In order to calculate the first and second order terms of the asymptotic expansion of \eqref{WpertypeII} we apply, respectively, the operators $\mathcal{L}_{\partial_\Om}$ and $\frac{1}{2}\mathcal{L}_{\partial_\Om}^2$ and evaluate the resulting expression at $\scri$. We also employ the notation introduced in Section \ref{sec4d}, where $\eqscr$ means that the expressions hold at $\scri$. Recall the definition $\Wperscr$   introduced in Section \ref{sec4d}
\begin{align}\label{FOWper}
 \fop \Wper & \eqscr  \Wperscr,
 \end{align}
   and let us also introduce the definition
   \begin{align}
\label{defF}
\frac{3}{2}\fop\left(\Wper\right)_{\alpha\beta}\y^\alpha \y^\beta  & \eqscr  \frac{3}{2}\Wperscr_{\alpha\beta}\y^\alpha \y^\beta  = :  \mathcal{F}.
\end{align}

Taking into account that both $\Wper_{\alpha\beta}$ and $ \Wppar_{\alpha\beta\nu}$ are $O (\Om)$ (cf. \eqref{eqexpWper} and \eqref{eqexpWperpar}),  equation \eqref{WpertypeII} gives at the order $\Omega$
\begin{align}
 \left(\fop\Wper\right)_{\mu'\nu'}  = & \frac{3}{2}\left(\fop\Wper\right)_{\alpha\beta}\y^\alpha \y^\beta \left(\y_{\mu'}y_{\nu'}  + \frac{\epsilon}{3} \gamma_{\mu'\nu'}\right)
 \\  + & 2 \epsilon \left(\fop\Wppar\right)_{\alpha\beta\nu}\y^\alpha \y^\beta \y_{(\mu'}\h^{\nu}{}_{\nu')}
 - \left(\fop\Wppar\right)_{(\mu|\alpha|\nu)}\y^\alpha\h^\mu{}_{\mu'}\h^\nu{}_{\nu'} + O(\Om).\label{der1}
\end{align}
The explicit expression for $\left(\fop\Wppar\right)_{\alpha\beta\nu}$ (and $\left(\fop\Wppar\right)_{(\mu|\alpha|\nu)}$) follows from \eqref{eqexpWperpar},
\begin{equation}
 \left(\fop \Wppar\right)_{\alpha\beta\nu} = -|\lambda|^{-1/2} (\Cot_{\Sigma})_{\alpha \beta  \nu}
\end{equation}
which in combination with expressions \eqref{FOWper} and \eqref{defF} takes \eqref{der1} into
\begin{align}
 \Wperscr_{\mu'\nu'}   \eqscr
  \mathcal{F} \left(\y_{\mu'}y_{\nu'}  + \frac{\epsilon}{3} \gamma_{\mu'\nu'}\right)
 -2\epsilon|\lambda|^{-1/2}  (\Cot_{\Sigma})_{\alpha\beta\nu}\y^\alpha \y^\beta \y_{(\mu'}\h^{\nu}{}_{\nu')}
 +|\lambda|^{-1/2} (\Cot_{\Sigma})_{(\mu|\alpha|\nu)}\y^\alpha\h^\mu{}_{\mu'}\h^\nu{}_{\nu'}. \label{WpertypeIIO1ap}
\end{align}
This is expression \eqref{WpertypeIIO1} used in the main text.

\bigskip

For the order $\Om^2$ we need some extra analysis. Recall that $\y$ is the unit vector obtained by proyecting the geodesic multiple WAND $\wand$ onto $\{\Om = const. \}$ submanifolds. Therefore, from \eqref{deckl},
\begin{align}
0 = \h^\beta{}_{\beta'} \wand^\alpha \nabla_\alpha\wand_\beta = \h^\beta{}_{\beta'} \wand^\alpha (\nabla_\alpha s) (u_\beta + y_\beta) +   \h^\beta{}_{\beta'} \wand^\alpha s (\nabla_\alpha  u_\beta + \nabla_\alpha  y_\beta)
 = \h^\beta{}_{\beta'} s^2 (u^\alpha + y^\alpha)  (\nabla_\alpha  u_\beta + \nabla_\alpha  y_\beta)
\end{align}
Taking into account that $\nabla_\alpha u_\beta = O(\Om)$ and $s = O(1)$ and different from zero at $\scri$, we have
\begin{align}\label{eqOOmy}
\h^\beta{}_{\beta'} (u^\alpha  \nabla_\alpha  y_\beta + y^\alpha \nabla_\alpha  y_\beta)   = O(\Om).
\end{align}
For the spacetime acceleration of the field $\y$ we get
\begin{align}\label{eqstacy}
\y^\alpha \nabla_\alpha  \y_\beta =  \y^\alpha D_\alpha  \y_\beta + \epsilon u^\mu  \y^\alpha (\nabla_\alpha  \y_\mu) u_\beta =  \y^\alpha D_\alpha  \y_\beta - \epsilon \y^\mu  \y^\alpha (\nabla_\alpha  u_\mu) u_\beta =  \y^\alpha D_\alpha  \y_\beta + O(\Om).
\end{align}
Denote $ \ay_\beta := \y^\alpha D_\alpha  \y_\beta$ to the intrinsic acceleration on $\{ \Om = const. \}$ leaves (note that this extends definition \ref{acyscri} off $\scri$.) Observe that from $\y$ being unit it follows $\ay_{\beta'} = \h^\beta{}_{\beta'} \ay_\beta$, so using \eqref{eqstacy}, equation \eqref{eqOOmy} can be rewritten
\begin{align}\label{eqOOmy0}
\h^\beta{}_{\beta'} u^\alpha  \nabla_\alpha  y_\beta + b_{\beta'}  = O(\Om).
\end{align}
On the other hand
\begin{align}
 \mathcal{L}_u \y_\beta = u^\alpha \nabla_\alpha \y_\beta  + \y^\alpha \nabla_\beta u_\alpha =  u^\alpha \nabla_\alpha \y_\beta + O(\Omega),
\end{align}
where note
 \begin{align}
  u^\alpha \nabla_\alpha \y_\beta & = \epsilon (\u^\mu u^\alpha \nabla_\alpha \y_\mu) \u_\beta-\epsilon (\y^\mu  u^\alpha \nabla_\alpha \y_\mu) \y_\beta + \h^\mu{}_{\beta} u^\alpha \nabla_\alpha \y_\mu = -\epsilon (\y^\mu u^\alpha \nabla_\alpha \u_\mu) u_\beta + \h^\mu{}_{\beta} u^\alpha \nabla_\alpha \y_\mu  \\ & = \h^\mu{}_{\beta} u^\alpha \nabla_\alpha \y_\mu + O(\Om).
\end{align}
Hence
\begin{align}
 \mathcal{L}_u \y_\beta = \h^\mu{}_{\beta} u^\alpha \nabla_\alpha \y_\mu + O(\Om),
\end{align}
so that equation \eqref{eqOOmy0} gives
\begin{align}\label{eqOOmy2}
\mathcal{L}_u y_\beta + b_\beta \eqscr 0
\end{align}
from which the order $\Omega$ of $\y$ is straightforward recalling that $ u = \epsilon |\lambda|^{1/2} \partial_\Om$
\begin{equation}\label{FOy}
  \fop \y_\beta \eqscr  -\epsilon |\lambda|^{-1/2} b_\beta.
\end{equation}
From this result, it is ready to obtain the leading order of $\h$, by just taking into account that
\begin{align}\label{FOh}
 O(\Om )=\mathcal{L}_u\gamma_{\alpha\beta} = -2 \epsilon  (\mathcal{L}_u \y)_{(\alpha} \y_{\beta)}+ \mathcal{L}_u \h_{\alpha\beta} \quad \Longrightarrow \quad  (\fop h)_{\alpha\beta} \eqscr  -2 |\lambda|^{-1/2}\ay_{(\alpha} \y_{\beta)}.
\end{align}
Note that expressions \eqref{FOy} and \eqref{FOh} also hold with upper indices  because
$\mathcal{L}_u\gamma^{\alpha\beta} = O(\Om)$.

\bigskip

We can now proceed with the calculation of the order $\Om^2$ of \eqref{WpertypeII}. Recall that $\Wper$ and $\Wppar$ are $O(\Om)$ (cf. \eqref{eqexpWper}, \eqref{eqexpWperpar}). Thus, $O(\Om^2)$ terms in \eqref{WpertypeII} must come from either $O(\Om^2)$ terms of  $\Wper$ or $\Wppar$ times $O(1)$ terms of the corresponding multiplying combination of $\y,\h,$ and $\gamma$; otherwise it must be a $O(\Om)$ term from $\Wper$ or $\Wppar$ times an $O(\Om)$ term of the corresponding multipliying combination of $\y,\h,$ and $\gamma$, times a factor of $2$ coming from the Leibniz rule.

By equation \eqref{eqexpWper}, the LHS of
\eqref{WpertypeII} gives
\begin{align}
 \frac{1}{2}\left(\fop^2\Wper\right)_{\mu'\nu'}  \eqscr \left( |\lambda|^{-1}\Sch_\Sigma^{(2)} - 4 |\lambda|g_{(4)} \right)_{\mu'\nu'} = -|\lambda|^{-1} (\Bac_\Sigma)_{\mu'\nu'},\label{2ndLHS}
\end{align}
where in the second equality we have introduced the expression \eqref{expresiong4} for $g_{(4)}$ with $\n = 4$.

We examine each term in the RHS of  separately \eqref{WpertypeII}. The first one
\begin{align}
&  \frac{1}{2} \fop^2\left(\frac{3}{2}\Wper_{\alpha\beta}\y^\alpha\y^\beta (\y_{\mu'}\y_{\nu'} + \frac{\epsilon}{3}\gamma_{\mu'\nu'})\right) = \\ & \frac{3}{4}\left(\fop^2\Wper\right)_{\alpha\beta}\y^\alpha\y^\beta (\y_{\mu'}\y_{\nu'} + \frac{\epsilon}{3}\gamma_{\mu'\nu'}) + \frac{3}{2}  \left(\fop\Wper\right)_{\alpha\beta}\fop\left(\y^\alpha\y^\beta (\y_{\mu'}\y_{\nu'} + \frac{\epsilon}{3}\gamma_{\mu'\nu'})\right) + O(\Om)=
\\ &
\left(\frac{3}{4}\left(\fop^2\Wper\right)_{\alpha\beta}\y^\alpha\y^\beta
+
  3 \left(\fop\Wper\right)_{\alpha\beta}\left(\fop\y\right)^{(\alpha}\y^{\beta)}
\right)
(\y_{\mu'}\y_{\nu'} + \frac{\epsilon}{3}\gamma_{\mu'\nu'})  + \frac{3}{2}  \left(\fop\Wper\right)_{\alpha\beta}\y^\alpha\y^\beta 2\left(\fop\y\right)_{(\mu'}\y_{\nu')}  + O(\Om),
\end{align}
where in the second equality we  have used $\fop\gamma_{\mu'\nu'} = O(\Om)$.
 Recalling the definitions  \eqref{FOWper}-\eqref{defF} and using the explicit expressions of $\fop^2 \Wper$ and  $\fop \y$ obtained in
 \eqref{2ndLHS} and \eqref{FOy} respectively, the above equation at $\scri$ is
\begin{align}
&   \frac{1}{2} \fop^2\left(\frac{3}{2}\Wper_{\alpha\beta}\y^\alpha\y^\beta (\y_{\mu'}\y_{\nu'} + \frac{\epsilon}{3}\gamma_{\mu'\nu'})\right) \eqscr \\ &   \left(-\frac{3}{2} |\lambda|^{-1}(\Bac_\Sigma)_{\alpha\beta}\y^\alpha \y^\beta
- 3 \epsilon |\lambda|^{-1/2} \Wperscr_{\alpha\beta}  \ay^{\alpha}\y^{\beta} \right)\left(\y_{\mu'}\y_{\nu'} + \frac{\epsilon}{3}\gamma_{\mu'\nu'}\right)
- 2 \epsilon \F   |\lambda|^{-1/2} \ay_{(\mu'}\y_{\nu')}.
\end{align}
By expression \eqref{WpertypeIIO1ap} we have
\begin{align}
  \Wperscr_{\alpha\beta}  \ay^{\alpha}\y^{\beta} = |\lambda|^{-1/2} (\Cot_{\Sigma})_{\alpha\beta\nu}\y^{\alpha}\y^{\beta}\ay^\nu,
\end{align}
so that
\begin{align}
&   \frac{1}{2} \fop^2\left(\frac{3}{2}\Wper_{\alpha\beta}\y^\alpha\y^\beta (\y_{\mu'}\y_{\nu'} + \frac{\epsilon}{3}\gamma_{\mu'\nu'})\right) \eqscr \\ &   \left(-\frac{3}{2}|\lambda|^{-1}(\Bac_\Sigma)_{\alpha\beta}\y^\alpha \y^\beta
- 3 \epsilon |\lambda|^{-1} (\Cot_{\Sigma})_{\alpha\beta\nu}\y^{\alpha}\y^{\beta}\ay^\nu \right)\left(\y_{\mu'}\y_{\nu'} + \frac{\epsilon}{3}\gamma_{\mu'\nu'}\right)
- 2 \epsilon \F   |\lambda|^{-1/2} \ay_{(\mu'}\y_{\nu')}.\label{2ndOterm1ap}
\end{align}

For the second term in the RHS of \eqref{WpertypeII} we use analgous arguments:
\begin{align}
& \frac{1}{2}\fop^2\left( 2 \epsilon \Wppar_{\alpha\beta\nu}\y^\alpha \y^\beta \y_{(\mu'}\h^{\nu}{}_{\nu')}\right) = \\
&
\epsilon \left(\fop^2  \Wppar\right)_{\alpha\beta\nu}\y^\alpha \y^\beta \y_{(\mu'}\h^{\nu}{}_{\nu')} \\ + &  2 \epsilon \left(\fop\Wppar\right)_{\alpha\beta\nu}\left(2 (\fop\y)^{(\alpha} \y^{\beta)} \y_{(\mu'}\h^{\nu}{}_{\nu')} +
 \y^{\alpha} \y^{\beta} (\fop\y)_{(\mu'}\h^{\nu}{}_{\nu')}
+
 \y^{\alpha} \y^{\beta} \y_{(\mu'}(\fop\h)^{\nu}{}_{\nu')}
\right) + O(\Om)
\end{align}
Now we replace $ \left(\fop^2  \Wppar\right)$ and  $ \left(\fop  \Wppar\right)$ from \eqref{eqexpWperpar}, $\fop\y$ from \eqref{FOy} and $\fop\h$ from \eqref{FOh}, so that we have
\begin{align}
& \frac{1}{2}\fop^2\left( 2 \epsilon \Wppar_{\alpha\beta\nu}\y^\alpha \y^\beta \y_{(\mu'}\h^{\nu}{}_{\nu')}\right) \eqscr \\
&  -4 |\lambda|^{-1/2}  D_{[\nu}\Wperscr_{\beta]\alpha}\y^\alpha \y^\beta\y_{(\mu'}\h^\nu{}_{\nu')}   - 2 \epsilon |\lambda|^{-1/2} (\Cot_\Sigma)_{\alpha\beta\nu} \left(-2  \epsilon |\lambda|^{-1/2} \ay^{(\alpha}\y^{\beta)}\y_{(\mu'}\h^\nu{}_{\nu')}- \epsilon |\lambda|^{-1/2} \y^\alpha \y^\beta  \h^{\nu}_{(\nu'} \ay_{\mu')} \right. \\ & \left. - |\lambda|^{-1/2}\y^\alpha \y^\beta \ay^\nu \y_{(\mu'}\y_{\nu')} - |\lambda|^{-1/2} \y^\alpha \y^\beta \y^\nu \y_{(\mu'}\ay_{\nu')}\right)
\\
= &  -4 |\lambda|^{-1/2}  D_{[\nu}\Wperscr_{\beta]\alpha}\y^\alpha \y^\beta\y_{(\mu'}\h^\nu{}_{\nu')}  + 2 |\lambda|^{-1} (\Cot_\Sigma)_{\alpha\beta\nu} \left(2 \ay^{(\alpha}\y^{\beta)}\y_{(\mu'}\h^\nu{}_{\nu')}+ \y^\alpha \y^\beta  \h^{\nu}_{(\nu'} \ay_{\mu')} + \epsilon \y^\alpha \y^\beta \ay^\nu \y_{\nu'}\y_{\mu'} \right)\\\label{2ndOterm2}
\end{align}
where in the last equality we use $(\Cot_\Sigma)_{\alpha\beta\nu} \y^\alpha \y^\beta \y^\nu \y_{(\mu'}\ay_{\nu')} = 0$ by the symmetries of the Cotton tensor.

For the last term in \eqref{WpertypeII} we obtain
\begin{align}
& \frac{1}{2} \fop^2\left(\Wppar_{(\mu|\alpha|\nu)}\y^\alpha\h^\mu{}_{\mu'}\h^\nu{}_{\nu'} \right) =
\\
& \frac{1}{2} \left(\fop^2\Wppar\right)_{(\mu|\alpha|\nu)}\y^\alpha\h^\mu{}_{\mu'}\h^\nu{}_{\nu'} +\left(\fop\Wppar\right)_{(\mu|\alpha|\nu)}\left( (\fop\y)^\alpha\h^\mu{}_{\mu'}\h^\nu{}_{\nu'}
+ \y^\alpha(\fop\h)^\mu{}_{\mu'}\h^\nu{}_{\nu'}
+\y^\alpha\h^\mu{}_{\mu'}(\fop\h)^\nu{}_{\nu'}
\right)
\end{align}
Again substituting  $ \left(\fop^2  \Wppar\right)$ and  $ \left(\fop  \Wppar\right)$, $\fop\y$ and $\fop\h$ by their corresponding expressions in \eqref{eqexpWperpar}, \eqref{FOy} and \eqref{FOh}, respectively, we get
\begin{align}
&  \frac{1}{2}\fop^2\left(\Wppar_{(\mu|\alpha|\nu)}\y^\alpha\h^\mu{}_{\mu'}\h^\nu{}_{\nu'} \right) \eqscr
\\ &
-\epsilon|\lambda|^{-1/2}\left(D_{[\nu} \Wperscr_{\alpha]\mu}+ D_{[\mu} \Wperscr_{\alpha]\nu}\right)\y^\alpha\h^\mu{}_{\mu'}\h^\nu{}_{\nu'}+ \epsilon|\lambda|^{-1}(\Cot_\Sigma)_{(\mu|\alpha|\nu)} \ay^\alpha h^{\mu}{}_{\mu'}  h^{\nu}{}_{\nu'}\\ & + |\lambda|^{-1}(\Cot_\Sigma)_{(\mu|\alpha|\nu)}\y^\alpha \left((\y^\mu \ay_{\mu'} + \y_{\mu'}\ay^\mu )\h^{\nu}{}_{\nu'}+ \h^{\mu}{}_{\mu'}(\y^\nu \ay_{\nu'} + \y_{\nu'} \ay^\nu) \right). \label{aux4}
\end{align}
We can rearrange the terms in the second line of \eqref{aux4} using the symmetries in the indices, namely,
\begin{align}
 \left(D_{[\nu} \Wperscr_{\alpha]\mu}+ D_{[\mu} \Wperscr_{\alpha]\nu}\right)\y^\alpha\h^\mu{}_{\mu'}\h^\nu{}_{\nu'} & = 2\left( D_{[\mu} \Wperscr_{\alpha]\nu}\right)\y^\alpha\h^\mu{}_{(\mu'}\h^\nu{}_{\nu')}, \\
 (\Cot_\Sigma)_{(\mu|\alpha|\nu)} \ay^\alpha h^{\mu}{}_{\mu'} h^{\nu}{}_{\nu'} & =  (\Cot_\Sigma)_{\mu\alpha\nu} \ay^\alpha h^{\mu}{}_{(\mu'}  h^{\nu}{}_{\nu')}.
\end{align}
Also notice that the terms in the last line of \eqref{aux4} can be also rearranged as follows
\begin{align}
& (\Cot_\Sigma)_{(\mu|\alpha|\nu)}\y^\alpha \left((\y^\mu \ay_{\mu'} + \y_{\mu'}\ay^\mu )\h^{\nu}{}_{\nu'}+ \h^{\mu}{}_{\mu'}(\y^\nu \ay_{\nu'} + \y_{\nu'} \ay^\nu) \right) \\ = &
(\Cot_\Sigma)_{\mu\alpha\nu}\y^\alpha \left(\y^{(\mu} \h^{\nu)}{}_{\nu'}\ay_{\mu'} + \y_{\mu'}\ay^{(\mu} \h^{\nu)}{}_{\nu'}+ \h^{(\mu}{}_{\mu'}\y^{\nu)} \ay_{\nu'} + \y_{\nu'} \h^{(\mu}{}_{\mu'} \ay^{\nu)} \right) \\ = & 2 (\Cot_\Sigma)_{\mu\alpha\nu}\y^\alpha \left(\y^{(\mu} \h^{\nu)}{}_{(\nu'}\ay_{\mu')} + \ay^{(\mu} \h^{\nu)}{}_{(\nu'}\y_{\mu')}\right).
\end{align}
%
Therefore
\begin{align}
&  \frac{1}{2}\fop^2\left(\Wppar_{(\mu|\alpha|\nu)}\y^\alpha\h^\mu{}_{\mu'}\h^\nu{}_{\nu'} \right) \eqscr \\
-& 2\epsilon|\lambda|^{-1/2}D_{[\mu} \Wperscr_{\alpha]\nu}\y^\alpha\h^\mu{}_{(\mu'}\h^\nu{}_{\nu')}+ |\lambda|^{-1}(\Cot_\Sigma)_{\mu\alpha\nu}\left( \epsilon \ay^\alpha h^{\mu}{}_{(\mu'}  h^{\nu}{}_{\nu')} + 2 \y^\alpha\y^{(\mu} \h^{\nu)}{}_{(\nu'}\ay_{\mu')} +2 \y^\alpha\ay^{(\mu} \h^{\nu)}{}_{(\nu'}\y_{\mu')}\right).\\ \label{2ndOterm3}
\end{align}
%
Replacing \eqref{2ndLHS} in the LHS of (the second order derivative of) \eqref{WpertypeII} and  \eqref{2ndOterm1ap},\eqref{2ndOterm2} and \eqref{2ndOterm3} in the RHS of (the second order derivative of) \eqref{WpertypeII} yields, after a multiplication by $-|\lambda|$,
\begin{align}
  (\Bac_\Sigma)_{\mu'\nu'}
  =& \left(\frac{3}{2}(\Bac_\Sigma)_{\alpha\beta}\y^\alpha \y^\beta+ 3 \epsilon  (\Cot_{\Sigma})_{\alpha\beta\nu}\y^{\alpha}\y^{\beta}\ay^\nu  \right)\left(\y_{\mu'}\y_{\nu'} + \frac{\epsilon}{3}\gamma_{\mu'\nu'}\right)
+ 2 \epsilon \F   |\lambda|^{1/2} \ay_{(\mu'}\y_{\nu')} \\
  + & 4 |\lambda|^{1/2}  D_{[\nu}\Wperscr_{\beta]\alpha}\y^\alpha \y^\beta\y_{(\mu'}\h^\nu{}_{\nu')}  - 2  (\Cot_\Sigma)_{\alpha\beta\nu} \left(2 \ay^{(\alpha}\y^{\beta)}\y_{(\mu'}\h^\nu{}_{\nu')}+ \y^\alpha \y^\beta  \h^{\nu}_{(\nu'} \ay_{\mu')} + \epsilon \y^\alpha \y^\beta \ay^\nu \y_{\nu'}\y_{\mu'} \right) \\
  -& 2\epsilon|\lambda|^{1/2}D_{[\mu} \Wperscr_{\alpha]\nu}\y^\alpha\h^\mu{}_{(\mu'}\h^\nu{}_{\nu')}+(\Cot_\Sigma)_{\mu\alpha\nu}\left( \epsilon \ay^\alpha h^{\mu}{}_{(\mu'}  h^{\nu}{}_{\nu')} + 2 \y^\alpha\y^{(\mu} \h^{\nu)}{}_{(\nu'}\ay_{\mu')} +2 \y^\alpha\ay^{(\mu} \h^{\nu)}{}_{(\nu'}\y_{\mu')}\right). \label{2ndOfullap}
\end{align}
Noting that $\y$ is orthogonal to $\ay$ and $\h$, it is immediate to check that contraction with $\y^{\mu'}\y^{\nu'}$ gives a tautology, meaning that $(\Bac_\Sigma)_{\mu'\nu'}\y^{\mu'}\y^{\nu'}$ is not determined by \eqref{2ndOfullap}. We thus define
\begin{equation}\label{defG}
 \mathcal{G} := \frac{3}{2} (\Bac_\Sigma)_{\alpha\beta}\y^{\alpha}\y^{\beta}.
\end{equation}
Also, using $\gamma_{\mu'\nu'} = -\epsilon\y_{\mu'}\y_{\nu'} + \h_{\mu'\nu'}$, we can rearrange the first term involving the Cotton tensor as
\begin{equation}\label{cotcomp1}
 3 \epsilon  (\Cot_{\Sigma})_{\alpha\beta\nu}\y^{\alpha}\y^{\beta}\ay^\nu  \left(\y_{\mu'}\y_{\nu'} + \frac{\epsilon}{3}\gamma_{\mu'\nu'}\right) =  2 \epsilon  (\Cot_{\Sigma})_{\alpha\beta\nu}\y^{\alpha}\y^{\beta}\ay^\nu  \y_{\mu'}\y_{\nu'} + (\Cot_{\Sigma})_{\alpha\beta\nu}\y^{\alpha}\y^{\beta}\ay^\nu \h_{\mu'\nu'},
\end{equation}
so inserting \eqref{defG} and \eqref{cotcomp1} into \eqref{2ndOfullap} yields
\begin{align}
  (\Bac_\Sigma)_{\mu'\nu'}
  =&~  \mathcal{G}\left(\y_{\mu'}\y_{\nu'} + \frac{\epsilon}{3}\gamma_{\mu'\nu'}\right) + 2 \epsilon \F   |\lambda|^{1/2} \ay_{(\mu'}\y_{\nu')}  +
    (\Cot_{\Sigma})_{\alpha\beta\nu}\y^{\alpha}\y^{\beta}\ay^\nu \h_{\mu'\nu'}  \\
  +& 4 |\lambda|^{1/2}  D_{[\nu}\Wperscr_{\beta]\alpha}\y^\alpha \y^\beta\y_{(\mu'}\h^\nu{}_{\nu')}  - 2  (\Cot_\Sigma)_{\alpha\beta\nu} \left(2 \ay^{(\alpha}\y^{\beta)}\y_{(\mu'}\h^\nu{}_{\nu')}+ \y^\alpha \y^\beta  \h^{\nu}_{(\nu'} \ay_{\mu')} \right) \\
 -& 2\epsilon|\lambda|^{1/2}D_{[\mu} \Wperscr_{\alpha]\nu}\y^\alpha\h^\mu{}_{(\mu'}\h^\nu{}_{\nu')}+(\Cot_\Sigma)_{\mu\alpha\nu}\left( \epsilon \ay^\alpha h^{\mu}{}_{(\mu'}  h^{\nu}{}_{\nu')} + 2 \y^\alpha\y^{(\mu} \h^{\nu)}{}_{(\nu'}\ay_{\mu')} +2 \y^\alpha\ay^{(\mu} \h^{\nu)}{}_{(\nu'}\y_{\mu')}\right). \label{2ndOfullap2}
\end{align}
Finally, reordering terms, we find
\begin{align}
  (\Bac_\Sigma)_{\mu'\nu'}
  =&  \mathcal{G}\left(\y_{\mu'}\y_{\nu'} + \frac{\epsilon}{3}\gamma_{\mu'\nu'}\right) + 2 \epsilon \F   |\lambda|^{1/2} \ay_{(\mu'}\y_{\nu')}  \\
  +  &   (\Cot_{\Sigma})_{\alpha\beta\nu}\left(\y^{\alpha}\y^{\beta}\ay^\nu \h_{\mu'\nu'} -  4 \ay^{(\alpha}\y^{\beta)}\y_{(\mu'}\h^\nu{}_{\nu')}-2\y^\alpha \y^\beta  \h^{\nu}{}_{(\nu'} \ay_{\mu')} \right) \\
 +& (\Cot_\Sigma)_{\mu\alpha\nu}\left( \epsilon \ay^\alpha h^{\mu}{}_{(\mu'}  h^{\nu}{}_{\nu')} +2  \y^\alpha\y^{(\mu} \h^{\nu)}{}_{(\nu'}\ay_{\mu')} +2 \y^\alpha\ay^{(\mu} \h^{\nu)}{}_{(\nu'}\y_{\mu')}\right)\\
 + & 4 |\lambda|^{1/2}  D_{[\nu}\Wperscr_{\beta]\alpha}\y^\alpha \y^\beta\y_{(\mu'}\h^\nu{}_{\nu')} - 2\epsilon|\lambda|^{1/2}D_{[\mu} \Wperscr_{\alpha]\nu}\y^\alpha\h^\mu{}_{(\mu'}\h^\nu{}_{\nu')} . \label{2ndOfullap3}
\end{align}
which is the expression \eqref{2ndOfull} used in the main text.

\end{document}